\documentclass[11pt]{article}
\usepackage{fullpage}
\usepackage{subfig}
\usepackage[usenames,dvipsnames]{xcolor}
\usepackage[linktocpage=true,colorlinks,citecolor=blue,linkcolor=BrickRed]{hyperref}
	\usepackage{latexsym,graphicx,epsfig,color}
\usepackage{amsfonts,amssymb,amsmath,amsthm,amstext}
\usepackage{enumitem}
\setitemize{itemsep=2pt,topsep=0pt,parsep=0pt}
\usepackage{url,setspace}
\usepackage{multirow}
\usepackage{rotating}
\usepackage{makeidx}
\usepackage{tikz}
\usepackage{accents}
\usepackage{xspace}
\usepackage{algorithm,algpseudocode}
\usepackage{bm}
\usepackage{changepage} 	
\usepackage{thmtools,thm-restate}

\newcommand{\IGNORE}[1]{}
\allowdisplaybreaks

\usetikzlibrary{decorations.markings}
\usetikzlibrary{arrows}
\tikzstyle{block}=[draw opacity=0.7,line width=1.4cm]
\tikzstyle{graphnode}=[circle, draw, fill=black!20, inner sep=0pt, minimum width=6pt]
\tikzstyle{point}=[circle, draw, fill=black!30, inner sep=0pt, minimum width=1pt]
\tikzstyle{input}=[rectangle, draw, fill=black!75,inner sep=3pt, inner ysep=3pt, minimum width=4pt]
\tikzstyle{unmatched}=[graphnode,fill=black!0]
\tikzstyle{shaded}=[graphnode,fill=black!20]
\tikzstyle{matched}=[graphnode,fill=black!100]  	
\tikzstyle{matching} = [ultra thick]
\tikzset{
    >=stealth',
    pil/.style={
           ->,
           thick,
           shorten <=2pt,
           shorten >=2pt,}
}
\tikzset{->-/.style={decoration={
  markings,
  mark=at position .5 with {\arrow{>}}},postaction={decorate}}}

\makeindex

\DeclareMathOperator{\argmax}{arg max}

\newtheorem{theorem}{Theorem}[section]
\newtheorem{claim}[theorem]{Claim}
\newtheorem{proposition}[theorem]{Proposition}
\newtheorem{lemma}[theorem]{Lemma}
\newtheorem{corollary}[theorem]{Corollary}

\newtheorem{observation}[theorem]{Observation}

\theoremstyle{definition}
\newtheorem{definition}[theorem]{Definition}
\newtheorem{example}[theorem]{Example}

\newtheorem{defn}[theorem]{Definition}

\usepackage{thmtools,thm-restate} 

\newcommand{\E}{\mathbb{E}}

\newcommand{\OPT}{\textsc{OPT}}

\newcommand{\Exp}[2]{\mathop{\mathbb E}_{#1}\displaylimits \left[#2\right]}

\def\eps {\epsilon}

\def \reals {\mathbb{R}}

\newcommand{\one}{\mathbf{1}\xspace}

\newcounter{note}[section]

\newcommand{\X}{\textbf{X}}
\newcommand{\Y}{\textbf{Y}}
\newcommand{\scol}{\ensuremath{s_{\text{col}}}}
\newcommand{\srow}{\ensuremath{s_{\text{row}}}}

\newtheorem{mainresult}{Main Result}

\title{Prophet Inequalities with Linear Correlations\\
and  Augmentations}

	\author{Nicole Immorlica\thanks{
	(nicimm@microsoft.com)
	Microsoft Research.
	}
	\and Sahil Singla\thanks{
        (singla@cs.princeton.edu)
        Princeton University and  
        Institute for Advanced Study. Supported in part by the
        Schmidt Foundation.
        }
	\and Bo Waggoner\thanks{
        (bwag@colorado.edu)
        University of Colorado.
        }
}

\date{ \today}

\begin{document}
\maketitle

\setlength{\abovedisplayskip}{2pt}
\setlength{\belowdisplayskip}{2pt}


\begin{abstract}{
\medskip

In a classical online decision problem, a decision-maker who is trying
to maximize her value inspects a sequence of arriving items to learn
their values (drawn from known distributions), and decides when to
stop the process by taking the current item. The goal is to prove a
``prophet inequality'': that she can do approximately as well as a
prophet with foreknowledge of all the values. In this work, we
investigate this problem when the values are allowed to be correlated.
Since non-trivial guarantees are impossible for arbitrary correlations, 
we consider a natural ``linear'' correlation structure introduced by Bateni et al.~\cite{BDHS-ESA15} as a generalization of the common-base value model of Chawla et al.~\cite{CMS-EC10Journal15}.

A key challenge is that
threshold-based algorithms, which are commonly used for prophet
inequalities, no longer guarantee good performance for linear correlations.
 We relate this
roadblock to another ``augmentations'' challenge that might be of 
independent interest: many existing prophet inequality algorithms are not robust to
 slight increase in the values of the arriving items.
We leverage this intuition to prove bounds
(matching up to constant factors) that decay gracefully with the
amount of correlation of the arriving items.
We extend these results to the case of selecting multiple items
by designing a new $(1+o(1))$ approximation ratio algorithm that is robust 
to augmentations. 

}\end{abstract}


\clearpage

\setcounter{tocdepth}{2}
\tableofcontents

\clearpage



\section{Introduction}\label{sec:intro}

In classic optimal-stopping problems, a decision-maker wishes to select one of a set  $[n] = \{1,\dots,n\}$ of options whose values are distributed according to a known joint distribution.  Option $i$ materializes at time $i$, revealing its value $X_i$.  The decision-maker must then either select option $i$, receiving a value of $X_i$, or permanently reject it and continue.  Her goal is to  choose an option whose value, in expectation over her selection algorithm and the randomness in the problem instance, obtains at least a $1/\alpha$-fraction of the expected maximum value for some approximation ratio $\alpha \geq 1$.  Such an approximation is referred to as a ``prophet inequality'' as it compares the decision-maker's performance to that of a prophet who knows the realizations of all values in advance and can always stop at the maximum. Examples of optimal-stopping problems include hiring, in which an employer interviews a sequence of candidates and must make a hiring decision on the spot; or house-buying in a sellers' market, in which a buyer must make an offer at the open house. 
These optimal-stopping problems became more popular in the last 15 years particularly because of their applications in mechanism design. E.g., an $\alpha$-prophet inequality often implies a posted-pricing mechanism that gets a $1/\alpha$-fraction of the maximum welfare for a sequence of bidders arriving online (see related work in Section~\ref{sec:related}). 




The approximability of optimal stopping problems depends heavily on the correlations of the option values.  In the case where all the values are independent, a tight ${2}$ approximation ratio was shown by Krengel and Sucheston~\cite{Krengel-Journal78}. In 1984, Samuel-Cahn~\cite{Samuel-Annals84} presented a simple \emph{median-of-maximum} ``threshold-based'' rule with the same performance: compute $\tau$ as the median of the distribution of the maximum value, and stop at the first $X_i$ exceeding $\tau$. Other threshold rules are also known to obtain ${2}$ approximation, e.g.,  Kleinberg and Weinberg~\cite{KW-STOC12} showed this for $\tau = \tfrac{1}{2}\Exp{}{\max_i X_i}$.
When the values are negatively correlated, the problem intuitively becomes even easier than the independent case: observing and rejecting a low $X_i$ increases the chances of seeing a high $X_{i'}$ in the future (and vice versa accepting a large $X_i$ decreases the chances of having missed a high $X_{i'}$ in the future). 
E.g., a simple implication of threshold-based algorithms is a ${2}$ approximation ratio for negatively associated  (a form of negative correlation) values;  see Appendix~\ref{sec:negCorr}.
Rinott and Samuel-Cahn~\cite{RS-Journal87,RS-Journal91} indeed show that the value of the optimal stopping algorithm for negative correlations exceeds that of the independent case, holding the marginals fixed. 
On the other hand, with general and positive correlation structures, no algorithm can guarantee better than $\Omega({n})$ approximation ratio (as is known from Hill and Kertz~\cite{HillKertz-Journal92}, and will also be a special case of our lower bounds).  Therefore, the question is how to impose a structure on the correlations that both models interesting applications and allows for better bounds.

\subsection{Linear Correlations Model}
We consider a \emph{linear correlations model} in which there exists a set of $m$ \emph{independent} variables $Y_1,\dots,Y_m$, with each option value $X_i$ being a positive linear combination of some subset:
  \[ \X = A \cdot \Y \]
where $A$ is a nonnegative matrix. The algorithm is given the matrix $A$ and the distributions of all the $Y_j$s, but when $X_i$ arrives, it only finds $X_i$ and not any of the realizations of the $Y_j$s. 
This model was introduced in an auctions context by Bateni et al.~\cite{BDHS-ESA15}, where it was inspired by the common-base value model of Chawla et al.~\cite{CMS-EC10Journal15}.
It has two natural parameters capturing the degree of correlation of an instance. If each row of $A$ has at most $\srow$ nonzero entries (\emph{row sparsity} $\srow$),  this implies that each $X_i$ only depends on at most $\srow$ different $Y_j$s. On the other hand, if each column of $A$ has at most $\scol$ nonzero entries (\emph{column sparsity} $\scol$), this implies that each $Y_j$ is only relevant to at most $\scol$ different $X_i$s.

\paragraph{General applications.}
Linear correlations occur in applications when each option $i$ is defined by the degree $A_{ij}$ to which it exhibits each feature $j\in[m]$.  The value $\sum_jA_{ij}Y_j$ of the option is then determined by the values $Y_j$ of its features, which is unknown to the decisionmaker.  Of particular interest in our setting are applications with many features, such as the hiring and house-buying examples often used to motivate prophet problems.  Other relevant applications include selecting hotel rooms, restaurants, and movies.  Here we elaborate on the hiring and house-buying examples, noting how they naturally exhibit column or row sparsity.\footnote{Note however that our bounds are expressed in terms of the minimum of row and column sparsity of an instance, and hence apply even to instances with high row/column sparsity.}  

In a hiring application, the features of a candidate might include where he received his education, his major, his work experience in each relevant industry, aspects of his personality, etc.  When the employer interviews a candidate, she learns how much she likes him, but not how to attribute her value for the candidate to particular features (every school/major/industry is a different feature). If the candidate pool is diverse, so that candidates come from many different schools/majors/industries, we might expect the instance to have a low column sparsity.  Similarly, in house-buying, the features of a house might include the commute time, the square footage, and various bells-and-whistles like the existence of a patio, a hot-tub, a roof-deck, the number of parking spaces, if any.  Again, the value of a house is a linear combination of the value of its features, but when seeing a house, the buyer may only be able to access and articulate an overall valuation.  If each house has a limited number of bells-and-whistles, we expect the instance to have low row-sparsity.

\paragraph{Mechanism design applications: Welfare for linearly correlated values.}
Prophet inequalities can directly imply social welfare and revenue guarantees for sequential posted-price mechanisms~\cite{CHMS-STOC10,KW-STOC12}.
In the simplest model, a single item is for sale to a sequence of arriving bidders with values $X_1,\dots,X_n$, drawn from distributions known to the seller.
A threshold-$\tau$ stopping rule immediately translates to a posted price $\tau$.
The item is purchased by the first bidder whose value satisfies $X_i \geq \tau$.
In particular, social welfare is the value of the bidder who purchases the item, so a prophet inequality directly translates to a social welfare guarantee.\footnote{Revenue guarantees, at least in the classical independent-$X_i$ model, can be obtained using a threshold in virtual value space.}
While we show that threshold-based policies fail for linearly correlated values, we obtain positive results with \emph{inclusion-threshold} policies.
These correspond to offering a fixed posted price to a predetermined subset of buyers, while the others are automatically rejected.

For linearly correlated bidder values, our positive results immediately imply social welfare guarantees using such inclusion-posted-price mechanisms.
Here, linear correlations capture some component of common values in bidder preferences.
Namely, there are different features $Y_1,\dots,Y_m$ of the object, with bidder $i$ placing weight $A_{ij}$ on feature $j$.
In the mechanism-design setting, it is particularly natural to make our assumption that the decisionmaker (here, the seller) is not able to access $Y_{j}$ when value $X_i$ arrives.



\paragraph{Results.}
We start from the observation that threshold-based algorithms cannot guarantee good approximation ratios as soon as \emph{any} correlations are introduced.
Therefore, we define \emph{inclusion-threshold} algorithms that probabilistically include a subset of the arrivals for consideration and take the first arrival \emph{in this subset} to exceed a threshold.

We first design an inclusion-threshold algorithm to obtain an $O(\scol)$ approximation ratio, i.e., a guarantee that degrades gracefully as the amount of correlation increases, from the known $O(1)$ bound for the independent case to the known $\Theta(n)$ worst-case bound.
Then, we design a more complex inclusion-threshold algorithm to obtain an $O({\srow})$ approximation ratio, i.e., another gracefully degrading guarantee.
Together, these prove an $O\big({\min\left\{{\srow}, {\scol}\right\}}\big)$ approximation guarantee for the linear correlations model.
We then design a lower bound instance and prove that this is tight up to constants, i.e., no algorithm can guarantee better than an $\Omega\big({\min\left\{{\srow}, {\scol}\right\}}\big)$ approximation.

\begin{mainresult}[Informal Theorem~\ref{thm:main}] 
For the linearly correlated prophet problem, there exists an $O({\min\{{\scol}, {\srow} \}})$ approximation ratio  algorithm.
\end{mainresult}

Finally, we extend these results to the case of selecting a subset of up to $r$ of the arriving options with a goal of maximizing their expected sum (also known as an \emph{$r$-uniform matroid} constraint). It is known that for independent distributions, $1+o(1)$ approximation ratio prophet inequality algorithms are possible for the case of large $r$~\cite{HKS-AAAI07,Alaei-FOCS11}. We show a similar result for linearly correlated instances with bounded column sparsity $\scol$. 

\begin{mainresult}[Informal Theorem~\ref{thm:matroidScolasmyp}] 
For the linearly correlated prophet problem where we select $r$ options, there exists a $1+o(1)$ approximation ratio  algorithm when $r \gg \scol$.
\end{mainresult}

The case of bounded row sparsity, however, turns out to be harder: regardless of $r$, no algorithm can guarantee better than an $\Omega({\srow})$ approximation ratio for unbounded $\scol$, as in the $r=1$ case.

\subsection{Techniques and the Augmented Prophets Problem} \label{sec:techniques}

A crucial technique for our results is to introduce and solve the \emph{augmented prophets} problem.
The idea is to suppose we have an instance with independent random variables $Z_1,\dots,Z_n$ and an algorithm, say a threshold rule, guaranteeing some approximation ratio.
Now suppose we ``augment'' the instance by sending instead the values $X_1 := Z_1+W_1, ~ \dots, ~ X_n := Z_n + W_n$ where the $W_i$s are nonnegative.
Does the algorithm (which does not know $W_i$s) continue to guarantee its original approximation ratio (measured against the maximum $Z_i$)?
One would hope so, as each arriving option has only increased, while the benchmark has not.
However, this turns out not to be true for the median-of-maximum threshold rule. 
E.g., if $Z_i \sim \text{Bernoulli}(p)$ i.i.d. for $p \ll \tfrac{1}{n}$, the median is zero, and augmenting the first arrival to a miniscule positive value causes the (strict) median threshold rule to always take it, resulting in an arbitrarily poor approximation.
Luckily, we show that the half-of-expected-maximum threshold algorithm retains its approximation guarantees, even when the $W_i$s are chosen by an adversary depending on the past $X_{i'}$s, i.e., $i'<i$.

Armed with this ``augmentation lemma'', we use subsampling to obtain inclusion sets of arrivals $\{X_i\}$ with significant independent portions $\{Z_i\}$.
In the bounded $\scol$ case, direct subsampling of arrivals succeeds.
The case where $\srow$ is bounded but $Y_j$ can appear in any number of arrivals is more challenging.
We show it suffices to obtain a contention-resolution style subsampling of the arrivals such that each $Y_j$ is well-represented, but only with its maximum coefficient $A_{ij}$.
We then use a graph-theoretic argument to construct such a scheme.

The augmented prophets problem is also our key technique for the linearly correlated prophets problem with an $r$-uniform matroid constraint. In this case, however, we notice that none of the existing $1+o(1)$ algorithms  are robust to augmentations. 
Hence we design a new $1+o(1)$ algorithm and prove its robustness using a much more sophisticated analysis involving a sequence of thresholds and ``buckets'' with different cardinality constraints.
By combining this augmentation result with random partitioning of the input, we obtain the $1+o(1)$ approximation for the $r$-uniform matroid problem with fixed $\scol$.

\begin{mainresult}[Informal Lemma~\ref{lemma:aug} and~\ref{lemma:r-aug}] 
For the augmented prophets problem, there exists a ${2}$ approximation ratio algorithm when selecting a single option and a $1+o(1)$ approximation ratio algorithm when selecting $r\gg 1$ options.
\end{mainresult}

 One can also view the augmented prophets problem as capturing correlations induced by a mischievous wish-granting genie who awards bonuses $W_i \geq 0$ at each step, but tries to choose them so as to worsen the algorithm's performance.  We think this problem is of independent interest and will find further applications in designing robust prophet inequality algorithms.

\subsection{Related Work} \label{sec:related}
The last decade has seen significant interest in prophet inequalities motivated by their applications in mechanism design.
Many works focus on multiple-choice prophet inequality problems.
This includes prophet inequalities for uniform matroids in~\cite{HKS-AAAI07,Alaei-FOCS11}, for general matroids in~\cite{CHMS-STOC10,Yan-SODA11,KW-STOC12,FSZ-SODA16,LeeS-ESA18}, for matchings and combinatorial auctions in~\cite{AHL-EC12,FGL-SODA15,DFKL-FOCS17,EHKS-SODA18}, and for arbitrary packing constraints in~\cite{Rubinstein-STOC16,RS-SODA17}. There has also been a lot of work on  variants of  prophet inequalities: the prophet secretary problem where the values arrive in a random order~\cite{EHLM-SIDMA17,AEEHKL-STOC17,CFHOV-EC17,EHKS-SODA18,ACK-EC18,LeeS-ESA18,CSZ-SODA19}, and the limited information setting where we only have sample-access to distributions~\cite{AKW-SODA14,CDFS-2018,RWW-ITCS20}. All these works assume mutually independent values, whereas capturing correlations and designing robust algorithms is the main challenge in our work.

 Rinott and Samuel-Cahn~\cite{RS-Journal87,RS-Journal91,RS-Journal92} study correlated prophet inequalities. However, their techniques are not applicable to our work because their positive results hold only for negatively correlated values. Furthermore, their benchmark is the expected maximum of independent values having the same marginal distributions. This benchmark  could be a  factor $n$ larger than the expected maximum for positively correlated values. 
 
 Our approach via the augmented prophets problem is also related to the line of work on designing robust stochastic optimization algorithms. Since  algorithms that assume \emph{known} input distributions tend to over-fit, here the goal is to design algorithms that are robust to adversarial noise~(see~\cite{Dia-tut,Moitra18,DiakonikolasKK016,LaiRV16,CharikarSV17,DiakonikolasKK018,EKM-TEAC18,LykourisML18,BGSZ-ITCS20} and references therein). Our single-item and multiple-items augmentation algorithms can be seen as robust prophet inequality algorithms that retain their guarantees even when the input distributions are augmented by an adversary. Another relevant reference is that of D{\"{u}}tting and Kesselheim~\cite{DuttingK-EC19}, which gives prophet inequalities  assuming only probability distributions that are $\epsilon$-close (in some metric) to the true distributions. Their technical results, however, are not useful here because augmented distributions can be very far from the original distributions.

\newcommand{\ALG}{\text{ALG}}

\section{Model and  Fixed-Threshold Algorithms} \label{sec:model}

\subsection{Model and Notation}
In the \emph{linear correlations} model, there are $n$ random variables $X_1, X_2, \ldots, X_n$ that linearly depend on $m$ independent nonnegative random variables (sometimes called \emph{features}) $Y_1, Y_2, \ldots, Y_m$ as
\[	\X = A \cdot \Y,
\]
where matrix $A$ only contains non-negative entries. Let $\srow$ denote the row-sparsity of $A$ (maximum number of nonzero entries in any row) and $\scol$ denote the column-sparsity of $A$ (maximum number of nonzero entries in any column).

An online algorithm is initially given $A$ and the distributions of $Y_1,\dots,Y_m$.
Then, it observes the realizations of $X_1,\dots,X_n$ one at a time.
After observing $X_i$, the algorithm decides either to stop and take the reward $X_i$, ending the process, or to reject $X_i$ and continue to $X_{i+1}$.
Given such an algorithm $\ALG$, we abuse notation by writing $\ALG$ for the reward of the algorithm, a random variable.
The algorithm has \emph{an approximation ratio of $\alpha$} for some $\alpha(n,\srow,\scol)$ if for all $n,\srow,\scol$ and all instances of the problem with these parameters,
  \[ \E[\ALG] \geq \frac{1}{\alpha(n, \srow, \scol)} \cdot \E [\max_i X_i] . \]
Such a guarantee is often called a \emph{prophet inequality} because it compares the algorithm to a ``prophet'' that can predict the realizations of all $X_i$ in advance and take $\max_i \{X_i \}$ every time.

We use the notation $(\cdot)^+$ to mean $\max\{\cdot, 0\}$.

Our examples frequently use random variables that are either zero or some fixed positive value.
We say the variable is \emph{active} if it takes its positive value.
We sometimes say that $X_i$ ``includes'' $Y_j$ if $A_{ij} > 0$.


\subsection{Fixed-Threshold Algorithms}
A \emph{fixed-threshold} algorithm selects a single threshold $\tau$ and takes the first arrival $X_i$ that exceeds $\tau$.
We refer to such an algorithm as $\ALG_{\tau}$.
Fixed-threshold algorithms have been very successful in prophet inequality design.
However, our first result shows their severe limitation for even mildly correlated prophet inequalities.

\begin{restatable}{lemma}{fixedThresh} \label{thm:fixedThreshold}
  In the linear correlations model, even for $\srow=\scol=2$ there exist instances where every fixed threshold $\tau$ algorithm $\ALG_{\tau}$ has an approximation ratio at least $\Omega({n})$.
\end{restatable}

The full proof is deferred to Appendix \ref{app:proofs}, but the instance is important and described next.
Intuitively, the problem is that cases where an arrival $X_i$ just crosses the threshold may be correlated with some later $X_{i'}$ being very large.
Taking $X_i$ prevents the algorithm from ever obtaining the gains from $X_{i'}$.
Our proof uses the following ``tower'' variables, which will also be useful later.

\begin{definition}[Tower $Y$ variables]
  Given $\epsilon>0$,  define the \emph{tower $\Y$ variables} as  $Y_i = \frac{1}{\epsilon^i}$ with probability $\epsilon^i$ and $Y_i = 0$ otherwise for each $i \in \{1,\dots,m\}$.
\end{definition}

\begin{example}[$\srow=\scol=2$ tower instance] \label{ex:2-tower}
  Take the tower $\Y$ variables.
  Let $A$ be an $n \times n$ matrix with entry $A_{i,i} = 1$ for all $i$ and $A_{i,i+1} = \epsilon$ for $i \in [1,n-1]$.  All other entries are $0$. Visually, 
  \begin{align*}
    X_1 = Y_1 + \epsilon Y_2,  \qquad     X_2 = Y_2 + \epsilon Y_3 ,   \qquad \ldots, \qquad
    X_n = Y_n .
  \end{align*}
\IGNORE{  \begin{align*}
    X_1 &= Y_1 + \epsilon Y_2  \\
    X_2 &= Y_2 + \epsilon Y_3  \\
    \ldots  \\
    X_n &= Y_n .
  \end{align*}
  }
We have $\scol = \srow =2$.  The point is that if $X_i$ is nonzero, then almost certainly $X_i = \frac{1}{\epsilon^i}$.
  But in this case, a threshold algorithm cannot distinguish between the more likely case that $Y_i = \frac{1}{\epsilon^i}$, in which case it should stop and take $X_i$, and the unlikely case that $Y_i = 0$ and $Y_{i+1} = \frac{1}{\epsilon^{i+1}}$, in which case it should wait and take $X_{i+1}$. 
  
  Indeed, these coefficients of matrix $A$ play an important role, and in Appendix~\ref{sec:unwtdLinCorr} we show that when entries of $A$ are restricted to being only $0$ or $1$, there exists a constant-factor approximation fixed-threshold algorithm. Roughly this happens because each $Y_j$ has ``limited influence'' on any $X_i$: either $Y_j$ appears with coefficient $0$ and has no influence on $X_i$, or it appears with coefficient $1$ and has the same influence on every such $X_i$.
  
 \begin{restatable}{theorem}{unwtdLinCorr}\label{thm:unwtdLinCorr}
The unweighted linear correlations problem has a fixed threshold constant-factor approximation algorithm.
\end{restatable}

\end{example}

This raises the question whether  for a general matrix $A$  any policy can achieve a better approximation, let alone a simple policy.
Our positive results will show that relatively simple \emph{inclusion-threshold} algorithms can achieve tight prophet inequalities.

\begin{definition} \label{defn:inclusionThreshold}
  An \emph{inclusion-threshold} algorithm selects a subset $S \subseteq \{1,\dots,n\}$ and threshold $\tau$, possibly both at random, and selects the first $X_i$ such that $i \in S$ and $X_i \geq \tau$.
  In other words, it commits to a subset $S$ of arrivals and applies a threshold policy to those $X_i$, ignoring the others.
\end{definition}

\section{Our Approach of Handling Correlations via Augmentations} \label{sec:aug}
In analysis of prophet inequalities, the problem is to upper-bound the expected maximum of the variables $X_i$ as compared to one's algorithm.
An important and common approach is to use the fact that for any threshold $\tau$,
\begin{align}
  \E [\max_i X_i] \quad \leq \quad  \tau + \E [\max_i \left(X_i - \tau\right)^+ ] \quad
  \leq \quad \tau + \sum_i \E [ \left(X_i - \tau\right)^+ ] \label{eqn:classic-max-bound}.
\end{align}
When all the arrivals $X_i$ are independent, it is known that one can always select $\tau$ such that the left and right sides of (\ref{eqn:classic-max-bound}) differ by at most a constant factor, i.e., $\frac{e}{e-1} \approx 1.6$ (this is related to the correlation gap~\cite{ADSY-OR12}).
In fact, the prototypical prophets analysis shows that setting some threshold $\tau$ allows $\ALG_{\tau}$ to approximate the right hand side up to a constant factor.
However, when $\{X_i\}$ are correlated, this could be a very loose upper bound. E.g., consider $X_1 = \cdots = X_n = Y_1 \sim \text{Bernoulli}(p)$.
The left side equals $p$ while the right side equals $\tau + np(1-\tau) = np + \tau(1-np) \geq np$ for $p < \frac{1}{n}$.
So the right side can be a factor $n$ larger than the left, and we cannot hope to approximate the right side with any algorithm.

One approach to correlated prophets \emph{could be} a direct analysis of the right-hand side of (\ref{eqn:classic-max-bound}) in cases of limited correlation.
Here, we take a different approach.
The first key idea is to use inclusion-threshold algorithms.
To see why, consider a first attempt: discard certain $X_i$ such that we are only left with a subset that are all independent of each other.
Now a standard prophet algorithm that only considers these $X_i$ would obtain a constant factor of the maximum in this subset.
One could then hope to argue that this subset's maximum approximates the original maximum up to a factor depending on the amount of correlation.
Indeed, one can show that this approach succeeds on the tower instance in Example~\ref{ex:2-tower} with $\srow = \scol = 2$, e.g., by including every other $X_i$.
But in general this approach cannot give tight bounds. This is because 
each $X_i$ contains $\srow$ variables, each of which can appear in up to $\scol-1$ other $X_{i'}$, so including $X_i$ requires eliminating $\approx \srow\scol$ other variables.
Our goal is to achieve approximations to within $\min\{\srow,\scol\}$ factors even if $\max\{\srow,\scol\} = n$.

So in addition to including only a subset of $X_i$, we will use a second key idea: decompose each variable as $X_i = Z_i + W_i$, where the $Z_i$s satisfy independence requirements and $W_i$s are viewed as ``bonus'' augmentations.
We will show that $\E[ \max_i Z_i]$ is an approximation to $\E[ \max_i X_i]$.
Then we will compete with $\E[ \max_i Z_i]$.
However, the augmentations add an additional challenge, requiring us to solve the following problem.

\begin{definition} [Single-Item Augmented Prophets Problem]
The algorithm is given the distributions of a set of independent nonnegative random variables $Z_1,\dots,Z_n$.
  Then, it observes one at a time the realizations of $X_i = Z_i + W_i$ for $i \in [1,n]$, where each $W_i$ is nonnegative and $Z_i$ is independent of $X_1,\dots,X_{i-1}$ for each $i$ (but $W_i$ could depend on $X_1, \ldots, X_{i}$).
  The algorithm chooses at each step whether to continue or stop and obtain value $X_i$.
  It must compete with $\E[\max_i Z_i]$.
\end{definition}
One can view the augmented prophets problem as capturing correlations induced by a mischievous  genie who awards bonuses $W_i \geq 0$ at each step so as to worsen the algorithm's performance.
We note that the genie cannot base her choices on the future, i.e., $W_i$ is a random variable that may be correlated with $Z_{i'}$ if $i' \leq i$ but not if $i' > i$.

Intuitively, it might seem like algorithms for prophets problems should continue to perform well, as the genie can only increase the rewards at each step.
However, this is not true for the classical median stopping rule, i.e, for $\tau = $ the median of $\max_i Z_i$ (see an example in Section~\ref{sec:techniques}).
Luckily,  a different threshold rule is robust:

\begin{lemma}[Single-item Augmentation lemma] \label{lemma:aug}
  For the augmented prophets problem, a fixed threshold algorithm with $\tau = \frac{1}{2} \E [\max_i Z_i]$ guarantees
$ \frac{\E [\ALG_{\tau}]}{\E [\max_i Z_i]} \geq \frac{1}{2} $.
\end{lemma}
\begin{proof}
  We ``augment'' a standard prophet inequality proof.
  Let $P = \Pr[ \max_i X_i \geq \tau]$. Now,
  \begin{align*}
    \E [\ALG_{\tau}] &=     P\cdot \tau + \sum_i \Pr[ X_{i'} < \tau ~ (\forall i' < i)] \cdot \E \big[ (X_i-\tau)^+  \mid X_{i'} < \tau ~ (\forall i' < i) \big]\\
                  &\geq P\cdot \tau + \sum_i (1-P) \cdot \E \big[ (X_i-\tau)^+  \mid X_{i'} < \tau ~ (\forall i' < i) \big]
\end{align*}
because  $\ALG_{\tau}$ selects no element with probability  $1-P$.
Nonnegativity of $W_i$ implies $X_i \geq  Z_i$, so
  \begin{align*}
   \E [\ALG_{\tau}]   &\geq P\cdot \tau + \sum_i (1-P) \cdot \E \big[ (Z_i-\tau)^+  \mid X_{i'} < \tau ~ (\forall i' < i) \big] \\
   & = P\cdot \tau +  (1-P) \cdot  \E \Big[ \sum_i  (Z_i-\tau)^+  \Big]
\end{align*}
because $Z_i$ is independent of the event $\{ X_{i'} < \tau ~ (\forall i' < i) \} $.  Since $\sum_i (Z_i - \tau)^+ \geq \max_i(Z_i - \tau)^+ $, 
  \begin{align*}
   \E [\ALG_{\tau}]   &\geq P\cdot \tau + (1-P) \cdot \E \big[\max_i (Z_i - \tau)^+\big]  \\
                  &\geq P\cdot \tau + (1-P)\cdot  \E\big[ \max_i Z_i - \tau \big]  
                  \quad = \quad    P\cdot \tau + (1-P) \tau   \quad =  \quad  \tau.
  \end{align*}
  This proves that $\E [\ALG_{\tau} ] \geq \frac{1}{2} \E [\max_i Z_i]$, as claimed.
\end{proof}

\paragraph{Multiple Items.}
The key idea in proving our $1+o(1)$ approximation ratio result for selecting multiple items for bounded $\scol$ is  to extend the  augmentation lemma to  cardinality constraints.

\begin{defn} [Multiple-Items Augmented Prophets Problem]
  In the \emph{augmented prophets problem with cardinality constraint $r$}, the algorithm is given the distributions of a set of independent nonnegative random variables $Z_1,\dots,Z_n$.
  Then, it observes one at a time the realizations of $X_i = Z_i + W_i$ for $i\in [1,n]$, where each $W_i$ is nonnegative and satisfies that each $Z_i$ is independent of $X_1,\dots,X_{i-1}$.
  The algorithm chooses at each step to reject or accept $X_i$ subject to taking at most $r$ variables total.
  It must compete with $\E\left[\sum_{i=1}^{r} Z^{(i)}\right]$, where $Z^{(i)}$ is the $i$th-largest $\{Z_1,\dots,Z_n\}$.
\end{defn}

Since none of the prior $1+o(1)$ approximation ratio algorithms for multiple-items is robust to augmentations, in Section~\ref{sec:multipItemsAug} we design a new algorithm to prove the following multiple-items augmentation lemma.
\begin{restatable}{lemma}{rAug}\emph{(Multiple-Items Augmentation Lemma)}. \label{lemma:r-aug}
  There is an algorithm for the augmented prophets problem with cardinality constraint $r$ achieving a $\Big(1 + O\big(\frac{(\log r)^{3/2}}{r^{1/4}}\big)\Big)$ approximation ratio.
\end{restatable}


Next, we utilize the single-item augmentation lemma, along with careful decompositions of $\{X_i\}$, to separately attack the single-item problem for the cases of bounded $\srow$ and $\scol$.




\section{Selecting a Single Item}
In this section we prove our main theorem.
Later, we will also  address cases where the algorithm can take multiple items.
\begin{restatable}{theorem}{UpperMain} \label{thm:main}
  There exists an inclusion-threshold algorithm for the linearly correlated prophet problem with approximation ratio $O({\min\{{\scol}, {\srow} \}})$.
\end{restatable}

  We will first show in Proposition~\ref{prop:colSparsity} that an inclusion-threshold algorithm guarantees $O({\scol})$; then in Proposition~\ref{prop:rowSparsity} that an inclusion-threshold algorithm achieves $O({\srow})$.   The algorithm that runs one of these according to which of $\scol,\srow$ is smaller is an inclusion-threshold algorithm achieving the claimed performance.

We will see that bounded column sparsity is the easier case, requiring a simpler algorithm and analysis.
For the case of bounded row sparsity, we will need much more careful reasoning about dependencies and correlations between $Y_j$.
This difficulty will also manifest quantitatively when we move to the cardinality-constraint setting in Section \ref{sec:multipleItems}, where better bounds will be achievable only in the case of bounded column sparsity.

\subsection{Bounded Column Sparsity}
Recall that column sparsity $\scol$ is the maximum number of times a given feature $Y_j$ may appear with nonzero coefficient.
We now give a relatively straightforward algorithm for achieving $\Omega\left(\frac{1}{\scol}\right)$ fraction of $\E[ \max_i X_i]$.
The idea is similar to the ``first attempt'' described in Section \ref{sec:aug}, using the single-item augmentation lemma (Lemma \ref{lemma:aug}) to overcome the challenges discussed there.

\begin{proposition}\label{prop:colSparsity}
  There exists an inclusion-threshold algorithm for the linearly correlated prophet problem with approximation ratio ${2e \cdot \scol}$.
\end{proposition}
\begin{proof}
  Choose $S \subseteq [n]$ by including each $i \in [n]$ independently with probability $\frac{1}{\scol}$.
  This gives the inclusion subset; now we define the threshold $\tau$.
  Assign each $Y_j$ to the first surviving $X_i$ that includes it, i.e., for each $i \in S$, construct a set $T_i := \big\{j : A_{ij} > 0 \text{ and } A_{i'j} = 0 ~ (\forall i' \in S \text{ where } i' < i) \big\}$.
  Let $Z_i = \sum_{j \in T_i} A_{ij} Y_j$ and set $\tau = \frac{1}{2} \E [\max\{ Z_1, \ldots, Z_n \}]$.
  If $i \not\in S$, then $T_i = \emptyset$ and $Z_i = 0$.

  By definition of an inclusion-threshold algorithm (Definition~\ref{defn:inclusionThreshold}), we automatically reject any $X_i$ such that $i \not\in S$, and we select the first arriving $X_i$ such that $i \in S$ and $X_i \geq \tau$.

  Now, by construction, we can write $X_i = Z_i + W_i$ where each $Z_i$ contains only variables $Y_j$ not appearing in any prior $X_{i'}$ for $i' \in S$ and $i' < i$.
  So $Z_i$ is independent of the preceding $X_{i'}$ under consideration.
  Hence by the single-item augmentation lemma (Lemma~\ref{lemma:aug}),
    $ \E [\ALG ] \geq \frac{1}{2} \E[ \max_i Z_i ]. $

  Next, we show that $\E[ \max_i Z_i ]$ is comparable to $\E[ \max_i X_i ]$.
  \begin{claim} \label{claim:col-z-approx}
    $\E [ \max\{Z_1,\ldots,Z_n\} ] \geq \frac{1}{e \cdot \scol} \E[ \max\{X_1,\ldots,X_n\} ]$, where
    the expectation is over the construction of $S$ as well as $Y_1,\dots,Y_n$.
  \end{claim}
  \begin{proof}[Proof of Claim \ref{claim:col-z-approx}]
    We prove that for every fixed realization of $Y_1,\dots,Y_n$, the inequality holds in expectation over $S$.
    Let $X_{i^*} = \max_i X_i$.
    For each $Y_j$ with $A_{i^*j} > 0$, we claim $\Pr[j \in T_i] \geq \frac{1}{e \cdot \scol}$ because $X_i$ survives with probability $\frac{1}{\scol}$ and independently, the other at most $\scol-1$ variables $X_{i'}$ with $A_{i'j} > 0$ all fail to survive with probability at least\footnote{We often use the inequality $\left(1-\frac{1}{N}\right)^{N-1} = \left(\frac{N-1}{N}\right)^{N-1} \geq \frac{1}{e}$, which follows from $\left(\frac{N}{N-1}\right)^{N-1} = \left(1 + \frac{1}{N-1}\right)^{N-1} \leq e$.} $\left(1 - \frac{1}{\scol}\right)^{\scol-1} \geq \frac{1}{e}$.
    (If $\scol=1$, then this probability is $1$.)
    So with probability only over the construction of $S$,
    \begin{align*}
      \E_S [Z_{i^*} ] \quad = \quad   \sum_j \Pr[j \in T_{i^*}] A_{i^* j} Y_j   \quad \geq \quad \frac{1}{e \cdot \scol} \sum_j A_{i^* j} Y_j  \quad = \quad    \frac{1}{e \cdot \scol} X_{i^*} .
    \end{align*}
    So we have $\E_S[ Z_{i^*}] \geq \frac{1}{e\cdot \scol}\max_i X_i$, so $\E_S[ \max_i Z_i ] \geq \frac{1}{e\cdot \scol} \max_i X_i$.
    This holds for each fixed realization of $Y_1,\dots,Y_n$, so it holds in expectation.
  \end{proof}
     Claim~\ref{claim:col-z-approx} completes the proof of Proposition \ref{prop:colSparsity}, as  we have
\[
    \E [\ALG ] \quad \geq \quad \frac{1}{2} \E [\max_i Z_i ]  \quad \geq \quad \frac{1}{2 e \cdot \scol} \E [ \max_i X_i ]. \qedhere
\]
\end{proof}

\subsection{Bounded Row Sparsity}
Recall that row sparsity $\srow$ implies that each $X_i$ only depends on at most $\srow$ different features $Y_j$; however, a given $Y_j$ may appear in arbitrarily many $X_i$s.
In this section, for notational convenience, we assume without loss of generality that $\max_i \{A_{ij} \} = 1$ for all $j$.
(If this is not the case, one can renormalize each column and redefine a scaled version of $Y_j$.)
We prove the following:
\begin{proposition}\label{prop:rowSparsity}
  There exists an inclusion-threshold algorithm for the linearly correlated prophets problem achieving approximation ratio ${2 e^3 \cdot \srow}$.
\end{proposition}
This case requires more care.
There does not seem to be an analogous approach to randomly excluding $X_i$, as for bounded column sparsity.
Moreover, an important observation is that the $X_i$ cannot be treated identically in a manner oblivious to the structure of $A$.
For every ``important'' row that ought to be included, there can be many unimportant rows.
Indeed, we can take any instance and prepend it with arbitrarily many variables of the form $X_i = Y_1$ without changing the row sparsity $\srow$.
An oblivious inclusion-threshold algorithm would essentially keep only variables from this prefix, ignoring the actual instance.

Before the formal proof of Proposition \ref{prop:rowSparsity}, we develop a tool to address this challenge.
The key idea is to design an inclusion scheme that, for any instance structure, allows each $Y_j$ to be both represented and ``independent'' with a reasonably high probability.
Here independence refers to not sharing an $X_i$ with any other included $Y_{j'}$.
Inspired by contention-resolution schemes, which have a similar flavor, we define a \emph{representative} construction of a subset of the $Y_j$ and corresponding $X_i$ with $A_{ij}=1$, where $X_i$ is matched to $Y_{j(i)}$.
\begin{definition}
  Consider a randomized selection of $S \subseteq \{1,\dots,n\}$ and $T \subseteq \{1,\dots,m\}$ of equal size with a perfect matching $j(i)$ satisfying $A_{ij(i)} = 1$.
  Call this construction \emph{$\alpha$-representative or $\alpha$-rep.} if 
\setlist{nolistsep}
\begin{enumerate}[label=(\roman*),noitemsep]
\item for all $j \in \{1,\dots,m\}$, we have $\Pr[j \in T] \geq \alpha$, and \label{prop1}
\item  for all $i,i' \in S$, $i\not= i'$, we have $A_{i'j(i)} = 0$. \label{prop2}
  \end{enumerate}
\end{definition}
Note that we cannot hope for better than an $O(\frac{1}{\srow})$-rep. construction, as any inclusion of some $Y_j$ can rule out $\srow-1$ other features.
This raises the question of whether one can achieve $\Omega(\frac{1}{\srow})$-rep.

\begin{lemma} \label{lem:rep}
  For any linearly correlated instance  there exists a $\frac{1}{e^2 \cdot \srow}$ rep. construction.
\end{lemma}
\begin{proof}
  For each $Y_j$, define its \emph{primary} $i(j)$ by picking any $i$ such that $A_{ij} = 1$ (by our renormalization assumption, there is at least one).
  Consider a directed graph $G$ where the nodes are $\{1,\dots,m\}$ representing the independent variables $Y_j$.
  There is a \emph{directed} edge $(j,j')$ if $A_{i(j)j'} > 0$, i.e., $j$ points to all other $j'$ who are included in its primary variable $X_{i(j)}$.
  We note that both edges $(j,j')$ and $(j',j)$ might be present.
  
  The key property is that each vertex in $G$ has out degree $\leq \srow-1$, as $Y_j$ has only one primary $X_{i(j)}$ and at most $\srow-1$ other variables $j'$ have $A_{i(j)j'} > 0$.
  Because average in-degree equals average out-degree, this implies there exists a vertex with in-degree at most $\srow-1$.
  Applying this argument recursively, we get the following claim.
  \begin{claim} \label{claim:PiOrder}
    There exists an order $\pi$ of the vertices of $G$ such that for every $j$, the induced subgraph on $\pi(1), \ldots, \pi(j)$ satisfies that the in-degree of $\pi(j)$ is at most $\srow-1$.
  \end{claim}
  \begin{proof}
    As shown, there exists some $j$ with in-degree at most $\srow-1$.
    Set $\pi(m) = j$.
    Now delete $j$ from the graph, including all edges to and from $j$.
    In this graph again all out-degrees are at most $\srow-1$, so we can recursively construct $\pi(m-1),\dots,\pi(1)$.
  \end{proof}
  Consider all the $Y_j$ variables in the order $\pi$ given by Claim~\ref{claim:PiOrder}.
  Initialize $S,T = \emptyset$.
  On considering $j$, if $j$ does not have an edge with any vertex $j' \in T$ (neither incoming nor outgoing), then independently with probability $\frac{1}{\srow}$, add $j$ to $T$ and add $j(i)$ (its primary variable) to $S$.
  With the remaining probability, ignore $j$ and continue.   We show that this  randomized construction of $S,T$ satisfies the  two properties of a $\frac{1}{e^2 \cdot \srow}$ rep. construction.

For Property~\ref{prop1}, note each $j$ has at most $2(\srow-1)$ total edges (both incoming and outgoing) to nodes $j'$ appearing prior to $j$ in the permutation: $j$ has at most $\srow-1$ outgoing edges in total, and by construction of $\pi$, has at most $\srow-1$ incoming edges from nodes prior to $j$ in $\pi$.
  So when we reach $j$ in the permutation, we consider it with probability at least the probability that all these $2(\srow-1)$ neighbors have been rejected, which is at least $\left(1 - \frac{1}{\srow}\right)^{2(\srow-1)} \geq \frac{1}{e^2}$; and then we include it with probability $\frac{1}{\srow}$.
  This shows that each $j$ is included with probability at least $\frac{1}{e^2 \cdot \srow}$.

  Next, for Property~\ref{prop2},  consider any $i,i' \in S$ with respective partners $j,j' \in T$.
  We must show $A_{ij'} = 0$.
  Note that by construction, $i = i(j)$ and $i' = i(j')$, i.e., they are the primary variables for $j,j'$.
  Either $j$ was selected into $T$ before or after $j'$.
  In either case, the second variable could only be selected if edge $(j,j')$ did not exist in the graph, which implies $A_{ij'} = 0$.
  
  This completes the proof of Lemma~\ref{lem:rep}.
\end{proof}

Given our representative construction, we are ready to complete the algorithm and proof.
\begin{proof}[Proof of Proposition \ref{prop:rowSparsity}]
  The algorithm is an inclusion-threshold algorithm.
  Its inclusion set $S$ is obtained by calling our representative construction of Lemma \ref{lem:rep}, which also produces a choice $j(i)$ for each $i \in S$.
  Define $Z_i = Y_{j(i)}$ for each $i \in S$ and $Z_i = 0$ if $i \not\in S$.
  Set $\tau = \frac{1}{2} \E[ \max_i Z_i]$.

  By the second property of representative constructions, $Z_i$ is independent of $X_{i'}$ for all $i' \in S, i' \neq i$.
  Therefore, by the Augmentation Lemma (Lemma \ref{lemma:aug}),
    \begin{align} \label{eq:augLemmRowSpars}
 \E [\ALG] \geq \frac{1}{2} \E [\max_i Z_i] .
     \end{align}
  Combining this with the following Lemma~\ref{claim:rowSparsity-z-approx} will prove Proposition \ref{prop:rowSparsity}.
  \begin{lemma} \label{claim:rowSparsity-z-approx}
    $\E[ \max_i Z_i ]\geq \frac{1}{e^3 \cdot \srow} \E[ \max_i X_i]$.
  \end{lemma}
  Before proving Lemma~\ref{claim:rowSparsity-z-approx}, we will need one more idea.
  Let $R_i = \{j : A_{ij} > 0\}$, the variables included in $X_i$.
  Notice that we may have $|R_i \cap T| \geq 2$, i.e., multiple variables $Y_j$ are members of $X_i$ and appear in the construction $T$.
  This can occur when $X_i$ is not primary for any of them.\footnote{An illuminating instance is: $X_i = Y_i$ for all $i \leq \srow$ and $X_{\srow + 1} = 0.99(Y_1 + \dots + Y_{\srow})$.}
  It will help to lower-bound the probability that $Y_j$ is the unique member of $R_i \cap T$.
  \begin{claim} \label{claim:rowSparsity-unique}
    Under the construction of Lemma \ref{lem:rep}, for each $j \in R_i$, $\Pr[R_i \cap T = \{j\}] \geq \frac{1}{e^3 \cdot \srow}$.
  \end{claim}
  \begin{proof}
    We have $\Pr[j \in T] \geq \frac{1}{e^2 \cdot \srow}$ by the representative construction.
    Meanwhile, conditioned on any other decisions, each $j'$ is included in $T$ with probability at most $\frac{1}{\srow}$, because it is considered with probability at most $1$ and included independently with probability $\frac{1}{\srow}$ conditioned on being considered.
    So $\Pr[R_i \cap T = \{j\} | j \in T] \geq \big(1 - \frac{1}{\srow}\big)^{\srow-1} \geq \frac{1}{e}$.
  \end{proof}
  \begin{proof}[Proof of Lemma~\ref{claim:rowSparsity-z-approx}]
    Fix the realizations of all $Y_j$.
    Let  $i^* = \argmax_i X_i$.
    First, notice that
        \begin{align*}
      \E [\max_i Z_i ] \quad &=  \quad  \E [ \max_{j \in T} \big\{ Y_j \big\} ] \quad \geq \quad \E \big[ \max_{j \in T \cap R_{i^*}} \big\{ A_{i^* j} \cdot Y_j \big\} \big] .
      \end{align*}

    Now the expected maximum of elements in $T \cap R_{i^*}$ is at least the sum over the elements of each's contribution to the max,   which is at least the chance it is the unique survivor times its value.  This allows us to relate a max to a sum, and it relies on the fact that the representative construction's randomness is independent of the realizations of the $\{Y_j\}$. Thus,
    \begin{align*}
      \E [\max_i Z_i ] ~ \geq  ~ \E \big[ \max_{j \in T \cap R_{i^*}} \big\{ A_{i^* j} \cdot Y_j \big\} \big] ~ \geq ~ \sum_j \Pr\big[T \cap R_{i^*} = \{j\} \big] \cdot A_{i^* j} Y_j  ~ \geq ~ \sum_j \frac{1}{e^3 \cdot \srow} A_{i^* j} Y_j ,
    \end{align*}
    where the last inequality uses Claim~\ref{claim:rowSparsity-unique}. Since $\sum_j A_{i^* j} Y_j  = X_{i^*}$, we have
    $  \E [\max_i Z_i ]     \geq     \frac{1}{e^3 \cdot \srow} X_{i^*}$ .
    Taking expectations over $Y_1,\dots,Y_n$ completes the proof.
  \end{proof}
  Finally, combining~\eqref{eq:augLemmRowSpars} with Lemma~\ref{claim:rowSparsity-z-approx} completes the proof of Proposition \ref{prop:rowSparsity}.
\end{proof}



\subsection{Lower Bounds} \label{sec:lowerBound}
We now give a matching hardness result, showing that no algorithm can do better than our results in the previous section up to constant factors.
\begin{example}[General tower instance] \label{ex:general-tower}
  Take the tower $\Y$ variables (recall $Y_i = \frac{1}{\epsilon^i}$ with probability $\epsilon^i$, and $0$ otherwise).
  Given input integer $c$, set $n = \srow = \scol = c$.
  Let $A$ be an $n \times n$ matrix with entry $A_{i,j} = 0$ for $j < i$ and $A_{i,j} = \epsilon^{j-i}$ for $j \geq i$.
  Visually,
  \begin{align*}
    X_1 = Y_1 + \epsilon Y_2 + ~ \cdots ~ \cdots ~ + \epsilon^n Y_n , \qquad 
    X_2 = Y_2 + \epsilon Y_3 + \cdots + \epsilon^{n-1} Y_n, \qquad
    \ldots , \qquad
    X_n = Y_n .
  \end{align*}
\IGNORE{  \begin{align*}
    X_1 &= Y_1 + \epsilon Y_2 + ~ \cdots ~ \cdots ~ + \epsilon^n Y_n  \\
    X_2 &= Y_2 + \epsilon Y_3 + \cdots + \epsilon^{n-1} Y_n  \\
    \ldots  \\
    X_n &= Y_n .
  \end{align*}
  }
  The difficulty here amplifies that of Example \ref{ex:2-tower}.
  If any of $Y_i,Y_{i+1},\dots,Y_n$ are active, then this will cause $X_i$ to be nonzero.
  Assuming only one of these variables is active (by far the most likely case), it is impossible for the algorithm to tell whether to stop or continue.
\end{example}

It will turn out that this instance is hard even if the algorithm is given additional power.
\begin{definition}
  In the \emph{fractional variant of the prophet problem}, at each arrival $i$, the algorithm may choose to take a fraction $p_i$ of the current arrival $X_i$ subject to always taking at most one  unit in total, i.e., $\sum_{i=1}^n p_i \leq 1$ with probability $1$.
  Its reward is $\sum_{i=1}^n p_i X_i$.
\end{definition}
One strategy available in the fractional prophet problem is to spend the entire budget on a single arrival, which is an algorithm for the standard prophets problem.
So a lower bound for the fractional problem immediately implies a lower bound for the prophets problem.

\begin{theorem} \label{thm:hardness-both}
  In the linearly correlated prophet problem, even if fractional, no online algorithm can guarantee a smaller approximation ratio than ${\min\{{\scol},{\srow}\}}$.
\end{theorem}
\begin{proof}
  We consider a family of instances of Example \ref{ex:general-tower} where $n = \scol = \srow$.
  Since for every $i$ we have $X_i \geq Y_i$, we get  that for sufficiently small $\epsilon$,
  \begin{align*}
    \E[\max_i\{X_i\}] \quad \geq \quad \E[\max_j\{Y_j\}]  
     \quad = \quad    \sum_{i=1}^n \epsilon^i \prod_{j>i} (1-\epsilon^j) \cdot \frac{1}{\epsilon^i} 
    \quad \geq \quad n (1-\epsilon)^n
    \quad \geq \quad n(1-n\epsilon).
  \end{align*}
  
  On the other hand, we show every online algorithm that may even fractionally select elements has value at most $1/(1-\epsilon)^2$.  This implies the approximation ratio can be  ${n}(1-n\epsilon)(1-\epsilon)^2$, which tends to $n$ as $\epsilon \rightarrow 0$.

\begin{restatable}{lemma}{TowerLowerBound} \label{lem:solveXiOne}
  Suppose at arrival $i$ we have $p = \sum_{i' < i} p_{i'}$ and $X_i = 1/\epsilon^i$.
  Conditioned on this event, any online algorithm obtains expected value from elements $i,\dots,n$  at most $({1-p})/{\epsilon^i}$.
  \end{restatable}
  
  Before proving Lemma~\ref{lem:solveXiOne}, we use it to prove that every algorithm has $O(1)$ expected value. Notice that if $X_i> 1/\epsilon^i$ then the online algorithm should never accept any fraction of the element $X_i$ as $X_{i+1}$ is guaranteed to be larger. Hence by Lemma~\ref{lem:solveXiOne}, the optimal algorithm $\ALG$ takes the smallest $i$ for which $X_i = 1/\epsilon^i$, which means
\begin{align*}
\E[\ALG] & \leq \sum_{i \geq 1} \Pr\big[ \big(X_{j} > 1/\epsilon^{j} \text{ for all } j<i \big) ~\&~  \big(X_i = 1/\epsilon^i \big)  \big] \cdot \frac{1}{\epsilon^i} \\
& = \sum_{i \geq 1}  \Pr[Y_{i-1} = 1/\epsilon^{i-1} ]  \cdot \big( \Pr[X_i = 1/\epsilon^i] \big) \cdot \frac{1}{\epsilon^i} \\
& \leq \sum_{i \geq 1}  \Pr[Y_{i-1} = 1/\epsilon^{i-1} ]  \cdot \big( \sum_{j\geq i} \Pr[ Y_j = 1/\epsilon^j] \big) \cdot \frac{1}{\epsilon^i} \\
& =  \sum_{i  \geq 1} \epsilon^{i-1}  \sum_{j \geq i}  \epsilon^j  \cdot \frac{1}{\epsilon^i} \quad  \leq   \quad \sum_{i  \geq 1}  \frac{\epsilon^{i-1} }{1-\epsilon}
\quad \leq \quad \frac{1}{(1-\epsilon)^2}.	\qedhere
\end{align*}
\end{proof}

Now we prove the missing lemma.

\begin{proof}[Proof of Lemma~\ref{lem:solveXiOne}]
We prove this lemma by reverse induction on $i$. It is immediately true for $i=n$, as spending the entire remaining budget $1-p$ on acquiring $X_i = \frac{1}{\epsilon^i}$ is optimal.
To prove the inductive step, notice the optimal online algorithm can be written as a convex combination of the following two algorithms: one that spends the entire remaining budget of $(1-p)$ on $X_i$ and another one that spends no budget on $X_i$ and plays optimally afterwards. We argue that the first algorithm is better, which means its expected value is $({1-p})/{\epsilon^i}$.

Observe that the second algorithm obtains nonnegative reward only if one of $Y_j$s for $j>i$ is active. In this case $X_{i+1} = 1/\epsilon^{i+1}$ (note it cannot be larger because $X_i = 1/\epsilon^i$), and hence by induction hypothesis the optimal online algorithm gets value $X_{i+1} = (1-p)/\epsilon^{i+1}$. Thus, the expected value of the algorithm is
\[	\Pr \big[\exists j>i \text{ s.t. } Y_j = 1/\epsilon^j \mid X_i = 1/\epsilon^i \big] \cdot (1-p)/\epsilon^{i+1}.
\]
We show $\Pr \big[\exists j>i \text{ s.t. } Y_j = 1/\epsilon^j \mid X_i = 1/\epsilon^i \big] \leq \epsilon$, which implies the first algorithm is always better. To see this, notice
\begin{align*}
\Pr\big[\exists j>i \text{ s.t. } Y_j = 1/\epsilon^j \mid X_i = 1/\epsilon^i \big] & = \frac{ \Pr\big[ \big( \exists j>i \text{ s.t. } Y_j = 1/\epsilon^j \big) ~\&~ \big( X_i = 1/\epsilon^i] \big) \big] } {\Pr[X_i = 1/\epsilon^i ]} \\
& =   \frac{\sum_{j>i} \Pr\big[ \big(  Y_j = 1/\epsilon^j \big) ~\&~ \big( X_i = 1/\epsilon^i] \big) \big] } {\sum_{j \geq i} \Pr\big[ \big(  Y_j = 1/\epsilon^j \big) ~\&~ \big( X_i = 1/\epsilon^i] \big) \big] } \\
& =   \frac{\sum_{j>i} \epsilon^j/(1- \epsilon^j)\cdot \prod_{j' \geq i} (1-\epsilon^{j'}) }{\sum_{j \geq i} \epsilon^j/(1- \epsilon^j)\cdot \prod_{j' \geq i } (1-\epsilon^{j'})} \quad  = \quad   \frac{\sum_{j>i} \epsilon^j/(1- \epsilon^j) }{\sum_{j \geq i} \epsilon^j/(1- \epsilon^j)}. 
\end{align*}
Now using $1+ x \leq \frac{1}{1-x} \leq 1+ 2x$ for $0 \leq x<0.5$, we get 
\begin{align*}
 \Pr\big[\exists j>i \text{ s.t. } Y_j  =  1/\epsilon^j \mid X_i = 1/\epsilon^i \big] & \leq  \frac{\sum_{j\geq i+1} \epsilon^j \cdot (1 +  2\epsilon^j) }{\sum_{j \geq i} \epsilon^j \cdot (1 +  \epsilon^j)} \\
 & = \frac{\epsilon^{i+1} - \epsilon^{n+1} + (2 \epsilon^{2i+2} - 2 \epsilon^{2n+2}) /(1+\epsilon) }{\epsilon^i -\epsilon^{n+1}  + (\epsilon^{2i}  -\epsilon^{2n+2})/(1+\epsilon)} \quad  \leq \quad  \epsilon,
\end{align*}
where the last inequality uses $\epsilon <1/2$ and $\epsilon^n < 2$.
\end{proof}


\section{Selecting Multiple Items} \label{sec:multipleItems}

In this section, we show that our approach via augmentations extends to a variant of the problem in which one may take up to $r \in \mathbb{N}$ of the arriving variables $X_1,\dots,X_n$.
Let $Q \subseteq \{1,\dots,n\}$ be the indices chosen by $\ALG$ with $|Q| \leq r$. Let $X^{(i)}$ denote the $i$th-largest realized variable (we later use notation $Z^{(i)}$ for the $i$th-largest among $\{Z_1,\dots,Z_n\}$ as well).
We have
  \[ \ALG = \sum_{i \in Q} X_i  
\quad \text{ while } \quad 
 \OPT = \sum_{i=1}^r X^{(i)} . \]

We refer to this problem as a cardinality constraint of $r$.
It is also referred to as selecting an independent set of a rank-$r$ matroid in the special case of \emph{$r$-uniform} matroids.

In this setting, there will be significant differences between row and column sparsity assumptions.
We will show that for bounded column sparsity $\scol$, one can design $(1+o(1))$-approximation algorithms for cardinalities $r \to \infty$, while this does \emph{not} hold for bounded row sparsity $\srow$.


\subsection{Bounded Column Sparsity} \label{subsec:uniform-scol}
As $r \to \infty$, we will show for bounded column sparsity an approximation ratio approaching $1$.
\begin{theorem} \label{thm:matroidScolasmyp}
For a fixed $\scol$,  the linearly correlated prophets problem with cardinality constraint $r$ admits a $\Big(1 + O\big((\frac{\scol}{r})^{1/5} (\log r)^{6/5}\big)\Big)$-approximation. 
\end{theorem}

The key idea is to prove an augmentation lemma  for selecting multiple items (restated below). 

\rAug*

In Section~\ref{sec:multipItemsAug} we prove this augmentation lemma, but before we use it to prove Theorem~\ref{thm:matroidScolasmyp}.


The idea of the reduction is that by randomly partitioning the variables into $\scol/\epsilon'$ ``groups'' gives us multiple independent Augmented Prophets problems. 
We think of each group as a subproblem of selecting $\epsilon' r/\scol$ elements and use the Augmentation Lemma to approximately solve it.

  \begin{proof}[Proof of Theorem~\ref{thm:matroidScolasmyp}]
Formally, let there be $c = \scol/\epsilon'$ sets, which we call \emph{groups}, $B_1,\dots,B_c$.
  For each $X_i$, place it in a group $j \in \{1,\dots,c\}$  chosen uniformly at random. For each $X_i$, let 
  \[ X_i' = \sum {A_{ij} Y_j \cdot \one[\text{$Y_j$ only appears once in the group containing $X_i$]}}
  \]
   denote the sum of $X_i$'s components $Y_j$ that do not appear with any other variable in the group containing $X_i$. Let  $\OPT_j'$ denote the sum of the largest $\epsilon' r/\scol$ elements $X_i'$ in group $B_j$.
  
  \begin{claim}\label{claim:bucketDecomp}
$	\E \big[ \sum_j \OPT_j' \big] \geq (1- \epsilon') \cdot \E[\OPT].$
  \end{claim}
  \begin{proof}
 Consider a fixed $X_i$. Condition on $X_i$ landing in a group. Notice that  each of its $Y_j$s have at least $1-\epsilon'$ chance of appearing only with $X_i$ in this group, and hence it  contributes to $\OPT_j'$.
  \end{proof}

  Since each group $B_j$ forms a separate instance of the Augmented Prophets problems, we can apply  the  Augmentation Lemma~\ref{lemma:r-aug} on each of them. Let $\ALG_j$ denote the algorithm's performance in group $j$. By selecting $\epsilon'$ less than   $O\big(\frac{(\log r)^{3/2}}{(\epsilon' r /\scol)^{1/4}} \big)$, i.e., choosing $\epsilon' = \Theta\big( (\frac{\scol}{r})^{1/5} (\log r)^{6/5}\big)$,
  we get
  \[	\sum_j \E [ \ALG_j ] \quad  \geq \quad \sum_j (1- O(\epsilon'))\cdot \E[\OPT'_j] \quad \geq \quad   (1- O(\epsilon')) \cdot \E[\OPT],
  \]
where the last inequality uses Claim~\ref{claim:bucketDecomp}.
  \end{proof}

\subsection{Bounded Row Sparsity} \label{sec:multipleRowSparse}

For cardinality constraints, the symmetry between bounds for row and column sparsity breaks: One cannot guarantee better than a $\Theta({\srow})$ approximation in general even as $r \to \infty$.
In fact, this follows by reducing to our previous hardness result for fractional prophets.
\begin{restatable}{theorem}{RowSparsHardness} \label{thm:RowSparsHardness}
  No algorithm for linearly correlated prophets with cardinality constraint can guarantee better than an $\Omega({\srow})$-approximation, even as $r \to \infty$.
\end{restatable}

Finally, we show that the $O({\srow})$-approximation  upper bound for single item can be extended to the cardinality constraint setting.
The intuition is straightforward, as we can simply instantiate $r$ parallel versions of our previous single item algorithm and assign arriving variables to each at random.
The analysis needs to show that no more than a constant factor is lost due to cases where members of OPT are sent to the same bucket.

\begin{restatable}{theorem}{RowSparsAlgo} \label{thm:RowSparsAlgo}
For all $r$, there is an $O({\srow})$-approximation for the linearly correlated prophet problem with cardinality constraint $r$.
\end{restatable}

We now present formal proofs of the last two theorems.
\begin{proof}[Proof of Theorem~\ref{thm:RowSparsHardness}]
  For any $r, \srow$, we construct the general tower instance of Example \ref{ex:general-tower} with parameter $\srow$, but we simply make $r$ copies of each variable $X_i$ (not independent, but exact copies).
  The row sparsity is unchanged.
  Now we will make the problem easier in two steps.

  (1) Allow variables to arrive in batches of size $r$.
  The algorithm may select any subset of the variables (until fulfilling its cardinality constraint), then reject the rest and receive the next batch.
  This is a strictly easier problem, so an algorithm's performance can only improve.
  Of course this will not help on this instance, since each batch of $r$ are all identical, and it will turn out to be optimal to either take them all or none.

  (2) Now instead we send the original tower instance (with no duplication), and we give the algorithm a cardinality constraint of $1$, but we allow it to pick a fractional amount of each variable $X_i$.
  In other words, we return exactly to the setting of Theorem \ref{thm:hardness-both}.
  In each case, the algorithm sees the same information before making each decision, i.e., the value of the current $X_i$.
  In the fractional problem, the algorithm can pick any fraction $p_i$ of $X_i$, so long as the total amount picked is at most one.
  In problem (1) above, the algorithm can pick any fraction $\frac{c}{r}$ of the variables that equal the original $X_i$, as long as it has does not exceed a total of $\frac{r}{r}$.
  The fractional problem allows more choice and the benchmarks (normalized) are the same, since the maximum of the duplicated instance will take all $r$ copies of the largest $X_i$.
  So the fractional problem is only easier.

  We now invoke Theorem \ref{thm:hardness-both}, which gives an lower bound of $\Omega({\srow})$-approximation on the fractional problem.
  (Note that $\srow$ did not change during the above reduction, although $\scol$ did.)
\end{proof}

\begin{proof}[Proof of Theorem~\ref{thm:RowSparsAlgo}]
  Given an instance, we create $r$ buckets $B_1,\dots,B_r$ and place each variable $X_i$ in a bucket $B_j$ uniformly at random.
  We then run our algorithm for the linearly correlated prophet problem with bounded $\srow$ in each bucket $j$ (call it $\ALG_j$), selecting one item.
  Given a fixed assignment of variables to buckets, in each bucket $B_j$ we have by the algorithm's guarantee that, with expectation over realizations of $\{X_i : i \in B_j\}$, by Proposition~\ref{prop:rowSparsity} 
    \[ \E_{\X}[\ALG_j] \geq \Omega\left(\frac{1}{\srow}\right) \cdot \E_{\X}[\max_{X_i\in B_j} X_i] . \]
  Taking expectations over both buckets and variables, we have
  \begin{align*}
    \E[\ALG] ~ = ~    \sum_{j=1}^r \E_B \E_{\X} \big[\ALG_j \big]  ~ \geq ~ \Omega\left(\frac{1}{\srow}\right) \sum_{j=1}^r \E_B \E_{\X}\big[\max_{X_i\in B_j} X_i\big]  
            ~=~    r \cdot \Omega\left(\frac{1}{\srow}\right) \E_{\X} \E_B\big[\max_{X_i\in B_1} X_i \big],
  \end{align*}
where the last equality is  by symmetry of the buckets.
  Now since $\OPT = X^{(1)} + \cdots + X^{(r)}$ where $X^{(i)}$ is the $i$th largest variable, and since for fixed variable realizations $\Pr[X^{(i)} = \max_{X_{i'}\in B_1} X_{i'}] = \frac{1}{r}\prod_{i'=1}^{i-1}(1- \frac{1}{r}) \geq \frac{1}{e\cdot r}$, we get
  \begin{align*}
    \E[\ALG] &\geq r \cdot \Omega\left(\frac{1}{\srow}\right) \E_{\X}\Big[\sum_{i=1}^r \Pr[X^{(i)} = \max_{X_{i'}\in B_1} X_{i'}] \cdot X^{(i)}\Big]  \\
            &\geq r \cdot \Omega\left(\frac{1}{\srow}\right) \E_{\X}\Big[\sum_{i=1}^r \frac{1}{e\cdot r} X^{(i)} \Big]   
        \quad    =   \quad   \Omega\left(\frac{1}{\srow}\right) .	\qedhere
  \end{align*}
\end{proof}

\subsection{Multiple-Items  Augmentation Lemma} \label{sec:multipItemsAug}

In Section~\ref{sec:aug} we showed a $2$-approximation single-item augmentation lemma using  the half of expected-maximum as a threshold. In this section, we prove a $1+o(1)$ approximation multiple-items augmentation lemma (Lemma~\ref{lemma:r-aug}), assuming the cardinality constraint $r$ is sufficiently large. 
We first give a surrogate benchmark $\OPT'$ that competes with $\OPT$.
We then define the algorithm and show that it competes with $\OPT'$.

\subsubsection*{Surrogate benchmark.}
The analysis hinges on  theshold $\tau_0 := \frac{\E[\OPT]}{\epsilon}$, where we call variables $Z_i$ and $X_i$ that fall above the threshold \emph{heavy} and below \emph{light}.
For light variables, we exclude the scenario where they are very small, below some threshold $\tau_c \leq \epsilon \frac{ \E[\OPT]}{r}$.
We will generally use the prime symbol $\prime$ to denote a version of a variable that is zeroed out if it's too heavy (or too light).
Let
\begin{align*}
  Z_i' = Z_i \cdot \one[\tau_c \leq Z_i < \tau_0]   \quad \text{and} \quad
  Z^{(i)\prime} = Z^{(i)} \cdot \one[\tau_c \leq Z_i < \tau_0].
\end{align*}
Let
\begin{align*}
  \OPT'   &= \OPT_1' + \OPT_2'  \qquad \text{where}\\
  \OPT_1' &= \max_i ~ Z_i \cdot \one[Z_i \geq \tau_0]  \quad \text{and} \quad 
  \OPT_2' = {\textstyle \sum_{j=1}^r Z^{(j)\prime}}.
\end{align*}
Note that $\OPT'$ only considers one heavy variable and throws all other heavy variables away.
Nevertheless, in Appendix~\ref{app:proofs} we will show the following claim that it competes with $\OPT$.

\begin{claim} \label{claim:onlyOneLarge}
  $\E [\OPT'] \geq (1-2\eps) \E [\OPT]$.
\end{claim}


\subsubsection*{Algorithm overview.}
For intuition, the problem with setting any particular fixed threshold is that variables with insignificant values of $Z_i$ can be boosted by the adversary to some $X_i$ just above the threshold.
The algorithm would use up its $r$ slots and be unable to take the later, larger arrivals that contribute to $\OPT$.

Therefore, we will define a sequence of thresholds, each a factor of $1-\epsilon$ apart.
The algorithm will have a certain number of slots $\tilde{r}_j$ for the ``bucket'' $j$ of arrivals between any two thresholds.
For this fixed bucket, such a boosting strategy by the adversary can only cost the algorithm a factor of $1-\epsilon$.

Roughly, this strategy will cover the case where $\OPT$ is concentrated.
To allow for cases where most of $\OPT$ comes from very rare, very large variables, we will also reserve a slot for such variables and analyze it separately.

\subsubsection*{Algorithm definition.}
We define a sequence of thresholds.
Recall that $\OPT = \sum_{j=1}^r Z^{(j)}$ where $Z^{(j)}$ is the $j$th-largest of $Z_1,\dots,Z_n$.
Let $c = \left\lceil \frac{1}{\epsilon} \ln \frac{r}{\epsilon^2} \right\rceil$ and define thresholds
\begin{align*}
  \tau_j &= (1-\epsilon)^j \frac{\E [\OPT]}{\epsilon}   & (j = 0,\dots,c)
\end{align*}

\begin{observation} \label{lemma:threshold-values}
  The largest threshold is  $\tau_0 = \frac{\E [\OPT]}{\epsilon}$ and the smallest threshold $\tau_c \leq \frac{\epsilon \cdot \E [\OPT]}{r}$.
\end{observation}
\begin{proof}
  $\tau_0$ is immediate, and we have $(1-\epsilon)^c \leq e^{-\epsilon c} \leq \exp\left( - \epsilon \frac{1}{\epsilon} \ln \frac{r}{\epsilon^2} \right) = \frac{\epsilon^2}{r}$.
\end{proof}

Now we define the size of each bucket.
Recall that we abuse notation by writing $Z_i \in \OPT$ if $Z_i$ is one of the $r$ variables included in the OPT solution.
Let
\begin{align*}
  r_j &= \E \left| \left\{ i : Z_i \in \OPT, \tau_j \leq Z_i \leq \tau_{j-1} \right\} \right|  & (j = 1,\dots,c)  \\
  \beta &= 3 \sqrt{r \ln(c/\epsilon)}  \\
  \tilde{r}_j &= r_j + \beta  & (j = 1,\dots,c)  \\
  \tilde{r}_0 &= 1 .
\end{align*}

We first define an algorithm $\ALG'$ that does not quite achieve the cardinality constraint $r$.
We will then modify it to obtain $\ALG$ with only a small loss in performance.
$\ALG'$ initializes $b_j = 0$ for $j=0,\dots,c$ and proceeds as follows when a variable $X_i$ arrives.
\begin{enumerate}[topsep=0mm,itemsep=0mm]
  \item If $X_i < \tau_c$, we discard $X_i$ and continue.
  \item Otherwise, let $j = \min \{j' : X_i \geq \tau_{j'} \}$.
  \item If $b_j < \tilde{r}_j$, we take $X_i$ and increment $b_j$.
        Otherwise (bucket $j$ is full), increment $j$ and repeat this step.
        If $j > c$, stop and discard $X_i$.
\end{enumerate}
In other words, we attempt to assign $X_i$ to its original bucket, but if that is full, we allow it to fall into buckets reserved for smaller variables (higher indices $j$).

Now, the final algorithm $\ALG$ is defined as follows: run $\ALG'$, but each time $\ALG'$ takes an arrival $X_i$, discard it independently with probability $\epsilon$.
If a variable is not discarded, but the cardinality constraint $r$ is reached, then discard it anyways.

\subsubsection*{Analysis.}
We show in Appendix~\ref{app:proofs} that $\ALG$ approximates $\ALG'$.
\begin{lemma} \label{lemma:alg-approx-algprime}
  For all $\epsilon \geq \frac{9 \left( \ln r \right)^{3/2}}{r^{1/4}}$, we have $\E[\ALG] \geq (1-2\epsilon) \E[\ALG']$.
\end{lemma}

Now, we analyze $\ALG'$.

Let us define the contributions of $\ALG'$ and $\OPT'$ bucket-by-bucket.
The top bucket of $\ALG'$ is split into cases where the corresponding $Z_i$ is heavy or light.
The following are random sets:
\begin{align*}
  O_j &= \{i \in \OPT' : \tau_j \leq Z_i < \tau_{j-1}\}  & (j = 1,\dots,c)  \\
  O_0 &= \{i \in \OPT' : \tau_0 \leq Z_i\}  \\
  B_j &= \{i \in \ALG  : \tau_j \leq X_i < \tau_{j-1}\}  & (j = 1,\dots,c)  \\
  B_0^{\text{light}} &= \{i \in \ALG'  : Z_i < \tau_0 \leq X_i\}  \\
  B_0^{\text{heavy}} &= \{i \in \ALG'  : \tau_0 \leq Z_i\} .
\end{align*}
We use the notation $i \in \OPT'$ to denote that $i$ contributes to $\OPT'$, i.e. either $Z_i$ is the largest among $\{Z_j\}$ or $Z_i'$ is among the $r$ largest of $\{Z_j'\}$.
Similarly, we write $i \in \ALG'$ to mean that the algorithm takes $X_i$.

Now, we break down $\ALG'$ as follows.
\begin{align*}
  \ALG'   &= \ALG_1' + \ALG_2'  \qquad \text{where}\\
  \ALG_1' = \sum_{i \in B_0^{\text{heavy}}} X_i & \quad \text{and} \quad
  \ALG_2' = \sum_{i \in B_0^{\text{light}}} X_i + \sum_{j=1}^c \sum_{i \in B_j} X_i .
\end{align*}
In other words, $\ALG_1'$ tracks the contribution of the special ``heavy'' bucket, but only in the case where the underlying variable $Z_i$ is heavy.
$\ALG_2'$ tracks the remaining case and all other buckets.
Notice these definitions are only for the purpose of analysis, as $Z_i$ is not observable to the algorithm.

Finally, we define
  \[ P := \Pr[ \max_i X_i \geq \tau_0 ] . \]
A key point will be that if $P$ is large, then the algorithm will often get some variable larger than $\tau_0$, which is good enough to compete with $\OPT$.

\begin{lemma} \label{lemma:alg-P-large}
  $\E[\ALG'] \geq \frac{P}{\epsilon}\E[\OPT]$.
\end{lemma}
\begin{proof}
  With probability $P$, some $X_i$ exceeds $\tau_0 = \frac{\E[\OPT]}{\epsilon}$.
  Since Bucket~$0$ is reserved for such arrivals with budget $b_0 = 1$, the algorithm gets  such $X_i$ if this occurs, so its expectation is $\geq P \tau_0$.
\end{proof}

Thus, if $P \geq \epsilon$, we are already done.
The rest of the analysis will leverage cases where $P$ is small.

\begin{lemma} \label{lemma:alg1-competes-opt1prime}
  $\E[\ALG_1'] \geq (1-P) \E[\OPT_1']$.
\end{lemma}
\begin{proof}
  \begin{align*}
    \E[\ALG_1']
    &=    \E \Big[\sum_i X_i \cdot \one[i \in B_0^{\text{heavy}}] \Big]  \quad \geq \quad \E \Big[ \sum_i Z_i \cdot \one[i \in B_0^{\text{heavy}}] \Big]  \\
    &=    \sum_i \Pr[\text{$b_0 = 0$ when $i$ arrives}] \cdot \E \left[ Z_i\cdot \one[Z_i \geq \tau_0] \right]  \qquad \text{(using independence)}  \\
    &\geq (1-P) \sum_i \E\left[ Z_i \cdot \one[Z_i \geq \tau_0] \right]  \\
 &  \geq (1-P) \E [ \max_i \{ Z_i \cdot \one[Z_i \geq \tau_0] \}] 
    ~~=~~    (1-P) \cdot \E[\OPT_1'] . \qedhere
  \end{align*}
\end{proof}

The final piece of the argument is to show that the ``buckets'' strategy works, i.e., it cannot be disrupted by augmentations.
The idea is that we have reserved an accurate number of slots in each bucket for the case where there is no augmentation.
An augmented variable $X_i$ can take away a bucket slot from some $Z_{i'}$, with $Z_{i'} \gg Z_i$, but then it will contribute about as much to $\ALG'$ as $Z_{i'}$ did to $\OPT_2'$.
In this case, we should be concerned that $\OPT_2'$ gets both $Z_i$ and $Z_{i'}$ while $\ALG'$ only gets $X_i$, with $X_{i'}$ disappearing thanks to the bucket being full.
However, the algorithm allows such an $X_{i'}$ to ``trickle down'' into a lower-tier bucket, in  particular, the slot that is not being used by $X_i$.
And if this slot is full as well, then in any case $\ALG'$ is competing with $\OPT_2'$.
\begin{lemma} \label{lemma:alg2-competes-opt2prime}
  For $\epsilon \geq \frac{9 \left( \ln r \right)^{3/2}}{r^{1/4}}$, we have $\E[\ALG_2'] \geq (1-\epsilon)^2 \E[\OPT_2']$.
\end{lemma}
\begin{proof}
  First, let $C$ denote the event that none of the $\OPT_2'$ buckets are filled to the $\tilde{r}_j$ capacities, i.e. $C$ is the event that $|O_j| \leq \tilde{r}_j$ for all $j=1,\dots,c$.
  \begin{claim} \label{claim:buckets-not-full}
    $\Pr[C] \geq 1-\epsilon$.
  \end{claim}
  \begin{proof}
    Recall that $r_j = \E | O_j|$.
    Because $|O_j|$ is a sum of independent Bernoulli random variables, we have by a standard Bernstein bound that
      \[ \Pr[ |O_j| - r_j \geq \sqrt{2 r_j t} + t] \leq e^{-t} . \]
    Setting $t = \ln\frac{c}{\epsilon}$, we get a bound of $\frac{\epsilon}{c}$; a union bound over the buckets will complete the proof.
    We just need to show that $\sqrt{2 r_j \ln \frac{c}{\epsilon}} + \ln\frac{c}{\epsilon} \leq \beta$, as then the probability of exceeding $r_j + \beta$ is only smaller.
    As in the proof of Claim \ref{claim:eps-small-discard}, for this choice of $\epsilon$, we have $\ln\frac{c}{\epsilon} \leq \ln(r)$, and $\ln(r) \leq r$, so
      \[ \textstyle \sqrt{2 r_j \ln \frac{c}{\epsilon}} + \ln \frac{c}{\epsilon} \leq \sqrt{2 r \ln \frac{c}{\epsilon}} + \sqrt{r \ln \frac{c}{\epsilon}} \leq \beta .  \qedhere \]
  \end{proof}
  Next we argue that at each tier of thresholds, $\ALG'$ is getting just as many variables as $\OPT_2'$, even if their identities are different.
  First, a helpful property:
  \begin{claim} \label{claim:given-C-more-alg2-notfull}
    Conditioned on $C$, suppose $|B_j| < \tilde{r}_j$ for some $j \geq 1$.
    Then $|B_0^{\text{light}}| + \sum_{j'=1}^j |B_{j'}| \geq \sum_{j'=1}^j |O_{j'}|$.
  \end{claim}
  \begin{proof}
    Consider any arrival $i$ that contributes to the right side.
    We claim it is also counted on the left.
    We know $X_i \geq Z_i \geq \tau_j$, because the variable contributes to the right side.
    So the algorithm will attempt to place $X_i$ in some assigned bucket $j' \leq j$.
    If it does not succeed because the bucket is full, it will proceed to $j'+1,\dots,$ and possibly eventually $j$.
    Because $|B_j| < \tilde{r}_j$, we know there is space for $i$ in bucket $j$, so the algorithm definitely takes $i$ in bucket $j$ or earlier.
    By definition, $Z_i < \tau_0$, so $i$ cannot be a member of $B_0^{\text{heavy}}$.
    Therefore, it is counted by the left side.
  \end{proof}
  Now we can show the key fact.
  \begin{claim} \label{claim:given-C-more-alg2}
    Conditioned on $C$, we have for all $j = 1,\dots,c$ that
      \[ \textstyle \left|B_0^{\text{light}}\right| + \sum_{j'=1}^j |B_{j'}| \geq \sum_{j'=1}^j |O_{j'}| . \]
  \end{claim}
  \begin{proof}
    By induction on $j$.
    For $j=1$, we must show $|B_0^{\text{light}}| + |B_1| \geq |O_1|$.
    Recall that an arrival $i$ is a member of $O_1$ if $\tau_1 \leq Z_i < \tau_0$.
    There are two cases.
    If $|B_1| = \tilde{r}_1$, i.e. the bucket is full, then the case is proven as we have assumed event $C$, which implies $|O_1| \leq \tilde{r}_1$.
    Otherwise, the case follows by Claim \ref{claim:given-C-more-alg2-notfull}.
    
    Now consider $j > 1$.
    If $|B_j| < \tilde{r}_j$, i.e. the bucket is not full, then the case follows by Claim \ref{claim:given-C-more-alg2-notfull}.
    Otherwise, i.e. bucket $j$ is full, then we have $|B_j| \geq |O_j|$ because of property $C$.
    Combining this with the induction hypothesis proves that $|B_0^{\text{light}}| + \sum_{j'=1}^j |B_{j'}| \geq \sum_{j'=1}^j |O_{j'}|$.
  \end{proof}

  We are now ready to complete the proof of Lemma \ref{lemma:alg2-competes-opt2prime}.
  From Claim \ref{claim:given-C-more-alg2}, given event $C$, we can make a one-to-one mapping from contributions $Z_i'$ of $\OPT_2'$ to contributions $X_j$ of $\ALG_2'$, such that $X_j$ is in the same bucket or a higher bucket than $Z_i'$.
  (I.e. map all elements of $O_1$ to elements of $B_1$ or $B_0^{\text{light}}$; map all elements of $O_2$ to remaining elements of these or to elements of $B_2$; and so on.)
  For each such pair, we have $X_j \geq (1-\epsilon) Z_i'$ because, at worst, both are in the same bucket.\footnote{For example, we may have $Z_i' \approx \tau_4$ while $X_j = \tau_3 = (1-\epsilon)\tau_4$.}
  In total, this implies that, conditioned on $C$, we always have $\ALG_2' \geq (1-\epsilon)\OPT_2'$; and $C$ occurs with probability at least $1-\epsilon$.
\end{proof}

\begin{corollary} \label{cor:algprime-competes-optprime}
  For $\epsilon \geq \frac{9 (\ln r)^{3/2}}{r^{1/4}}$, we have $\E[\ALG'] \geq (1-\epsilon)^2 \cdot  \E[\OPT']$.
\end{corollary}
\begin{proof}
  If $P \geq \epsilon$, then by Lemma \ref{lemma:alg-P-large}, we have $\E[\ALG'] \geq \E[\OPT] \geq \E[\OPT']$, proving the claim.
  Otherwise, by Lemma \ref{lemma:alg1-competes-opt1prime}, $\E[\ALG_1'] \geq (1-P) \E[\OPT_1'] \geq (1-\epsilon)\E[\OPT_1']$; and by Lemma \ref{lemma:alg2-competes-opt2prime}, $\E[\ALG_2'] \geq (1-\epsilon)^2 \E[\OPT_2']$.
  So in the case $P < \epsilon$, we have
  \begin{align*}
    \E[\ALG'] \quad = \quad \E[\ALG_1'] + \E[\ALG_2']  \quad  &\geq \quad (1-\epsilon)\E[\OPT_1'] + (1-\epsilon)^2\E[\OPT_2']  \\
              &\geq \quad (1-\epsilon)^2 \E[\OPT'].	\qedhere
  \end{align*}
\end{proof}

\begin{proof}[Proof of Augmentation Lemma~\ref{lemma:r-aug}]
  For $\epsilon \leq \frac{1}{2}$, $\epsilon \geq \frac{9 (\ln r)^{3/2}}{r^{1/4}}$, we have:
  \begin{align*}
    \E[\ALG] \stackrel{\text{Lemma~\ref{lemma:alg-approx-algprime}}}{\geq} (1-2\epsilon)\E[\ALG']     \stackrel{\text{Corollary \ref{cor:algprime-competes-optprime}}}{\geq} (1-2\epsilon)^3 \E[\OPT']  
    \stackrel{\text{Claim \ref{claim:onlyOneLarge}}}{\geq} (1-2\epsilon)^4 \E[\OPT]  .
  \end{align*}
  So we obtain a $\big(1+O(\epsilon)\big)$-approximation.
\end{proof}

\medskip
\noindent
{\bf Acknowledgments}.
We are thankful to the anonymous reviewers of EC 2020 for helpful comments on improving the presentation of this paper.

\appendix


\section{Unweighted Linear Correlations} \label{sec:unwtdLinCorr}
In this section we consider a linear correlations model where the nonnegative matrix $A$ from $ \X = A \cdot \Y $ is unweighted, i.e., each of its entry is either $0$ or $1$. Alternately, for $i\in [n]$ there are known sets $S_1, S_2, \ldots, S_n \subseteq [m]$ such that $X_i = \sum_{j\in S_i }Y_j$.
Our lower bound tower instance from Section~\ref{sec:lowerBound} no longer holds as it crucially expolits that matrix $A$ has entries that decrease exponentially in $\eps$. Can we  do better than a $\Theta\left({\min\{\scol,\srow\}}\right)$ approximation ratio? 
One might wonder if there exists an alternate hardness instance that only has $0-1$ entries in $A$. We show that this is not the case. In fact, there exist simple threshold-based constant approximation algorithms.

\unwtdLinCorr*

The main intuition in the proof of Theorem~\ref{thm:unwtdLinCorr} is that for the unweighted problem each \emph{independent} $Y_j$ has limited ``influence'' on the $X_i$s. This is because  either $Y_j$ appears with  coefficient $0$, in which case it has no influence on the value $X_i$, or it appears with coefficient $1$, in which case it has the same influence in the value of every such $X_i$.  A threshold algorithm is therefore difficult to fool because unlike the tower instance, it is not possible to have a scenario where a $Y_j$ is very large but our algorithm selects it within an $X_i$ where it appears with a small coefficient $\eps$. 

For readers familiar with the revenue-maximization result of Babaioff et al.~\cite{BILW-FOCS14}, i.e., the best of selling items individually and selling all the items together in a single bundle is a constant factor approximation to optimal revenue, our result has a similar flavor, although the technical details are quite different. We decompose our problem instance into a ``core'' and a ``tail'' part. The tail consists of cases where any $Y_j$ exceeds a boundary $\tau$; the core, the rest. For the tail case we show that approximating $\E[\max_j\{Y_j\}]$ (the best individual item) suffices, and in the core case we can approximate the best bundle $X_i$.

This argument will show that there is one fixed threshold $\tau_{\text{core}}$ such that the algorithm taking the first arrival above $\tau_{\text{core}}$ achieves a constant approximation to the optimal core contribution to $\max\{X_i\}$; and similarly for $\tau_{\text{tail}}$ and the tail part.
There remains a corner case, where we show a fixed threshold equal to the boundary $\tau$ gives a constant factor.
Thus, for any given instance, one can select  among $\tau, \tau_{\text{core}}, \tau_{\text{tail}}$ to get a fixed-threshold constant-factor approximation algorithm.


\subsection{Notation and Proof Overview}
We first choose a real number $\tau$ representing a boundary.
Let $p_j = \Pr[Y_j >\tau]$. We set $\tau$ such that $\prod_j (1-p_j) =1/2$, i.e., with half probability all $Y_j$ are below $\tau$.
We let the set $A$ be all ``heavy'' $Y_j$ variables: $A := \{j : Y_j > \tau\}$.

Recall that for each $X_i$, the set of active indices is $S_i$, so we have $X_i = \sum_{j \in S_i} Y_j$.
We first upper-bound $\OPT$ by contributions from the \emph{core} event that $A = \emptyset$ (all $Y_j$ are small) and from the remaining \emph{tail} event.

\begin{claim} \label{claim:unwtdBoundOPT}
  \[  \E[\OPT] \leq  \E[\max_{i} X_i \mid A=\emptyset] + \sum_{j} p_j \cdot \E[Y_j \mid Y_j > \tau] . \]
\end{claim}
\begin{proof}
  For any outcome of $Y_j$s, we separately count those in $A$ (larger than $\tau$) and the rest, and relax the objective to always take the large ones:
  \begin{align*}	
    \max_i X_i  \quad= \quad \max_{i} \{ \sum_{j \in S_i}Y_j \} 
    \quad \leq \quad  \max_{i} \{ \sum_{j \in S_i \cap \overline{A}}Y_j  \} + \sum_{j\in A} Y_j.
  \end{align*}
  Now taking expectations on both sides,
  \begin{align*}
  \E[\OPT] \quad= \quad 	 \E[\max_i X_i ] \quad & \leq \quad  \E[\max_{i} \{ \sum_{j \in S_i \cap \overline{A}}Y_j  \} ]+ \E[\sum_{j\in A} Y_j]\\
  	 &  \leq \quad \E[\max_{i} X_i \mid A = \emptyset ] +  \sum_{j} p_j \cdot   \E[Y_j \mid Y_j > \tau].
  \end{align*}
  To justify that $\E[\max_{i} \{ \sum_{j \in S_i \cap \overline{A}}Y_j  \} ] \leq \E[\max_{i} X_i \mid A = \emptyset ]$, use a coupling argument: First draw all $Y_j$ from their initial distributions and consider the value of $\max_i \{ \sum_{j \in S_i \cap \overline{A}} Y_j\}$. Now take any variables $Y_j > \tau$, and redraw them until they fall below $\tau$.
  The value of the inner sum can only increase, but now we are exactly obtaining $\E[\max_i X_i \mid A = \emptyset]$.
\end{proof}

Given that Claim \ref{claim:unwtdBoundOPT} upper-bounds $\OPT$ by the sum of two terms, our proof goes in two steps.
First, we approximate the \emph{tail} contributions $\sum_{j} p_j \cdot   \E[Y_j \mid Y_j > \tau]$.
In Lemma~\ref{lem:unwtdMaxYj} we use the Augmentation Lemma~\ref{lemma:aug} to 
give a simple fixed threshold-$\tau_{\text{tail}}$ algorithm with expected value $\Omega(\max\{Y_j\})$. This is a bit surprising because the lower bound in Section~\ref{sec:lowerBound} actually proves such a result is not possible for \emph{weighted} linear correlations. In Claim~\ref{claim:unwtdMaxYj}, we show that $\Omega(\max\{Y_j\})$ suffices to capture the tail term.

To capture the \emph{core} contributions $\E[\max_{i} X_i \mid A=\emptyset]$,  in Claim~\ref{claim:unwtdApplyConcent} we argue that for all instances where $Y_j$s are bounded by $\tau$, we can use concentration of XOS functions to argue that $\Pr[\max_i X_i > 1/2 \cdot \E[\max_i X_i]]$ is at least a constant, and hence a simple fixed threshold-$\tau_{\text{core}}$ algorithm suffices.  
This second step also holds for prophets with weighted linear correlations.
There is also a corner case where $\tau$ is too large to apply concentration.
But in this case, setting a threshold $\tau$ will directly achieve a constant factor.

\subsection{Proof}

\begin{lemma} \label{lem:unwtdMaxYj}
For the prophet inequality problem with unweighted linear correlations, there exists a fixed threshold algorithm with expected value $\Omega(\max\{Y_j\})$.
\end{lemma}
\begin{proof}
Define $Z_i$ to be the sum of $Y_j$s that appear in $X_i$ and have not appeared in any $X_{i'}$ for $i' < i$. Since every $Y_j$ appears in some $Z_i$, we know $\max\{Z_i\} \geq \max\{Y_j\}$.  The Augmentation Lemma~\ref{lemma:aug} now completes the proof.
\end{proof}

Now we argue that $\E[\max_j Y_j]$ takes care of the second term in Claim~\ref{claim:unwtdBoundOPT}.
\begin{claim} \label{claim:unwtdMaxYj}
\[	\E[\max_j Y_j] \geq \frac{1}{2} \cdot \sum_{j} p_j \cdot   \E[Y_j \mid Y_j > \tau].
\]
\end{claim}
\begin{proof}  We have
\begin{align*}
\E[\max_j Y_j] ~~\geq~~ \sum_{j} \Pr[A = \{j\}] \cdot \E[\max_j Y_j \mid A = \{j\}] ~~\geq~~ \sum_j \frac{p_j}{2} \cdot \E[ Y_j \mid A = \{j\}],
\end{align*}
where the second inequality uses $\Pr[A = \{j\}] = \Pr[Y_j > \tau] \Pr[Y_{j'} \leq \tau (\forall j' \neq j)] \geq (p_j)\left(\frac{1}{2}\right)$ and that $\max_j Y_j$ given that $A = \{j\}$ is the same as $Y_j$.
\end{proof}

\begin{corollary} \label{cor:unwtdTail}
  For any instance with unweighted linear correlations, there exists $\tau_{\text{tail}}$ such that the algorithm setting a fixed threshold of $\tau_{\text{tail}}$ obtains $\E[\ALG] \geq \Omega\left( \sum_j p_j \cdot \E[Y_j \mid Y_j > \tau] \right) $.
\end{corollary}

We now turn to the core portion of contributions to $\OPT$.

\begin{claim} \label{claim:unwtdApplyConcent}
  Let $V = \E[\max_i X_i \mid A=\emptyset]$.
  If the boundary satisfies $\tau \leq V /10$, there exists $\tau_{\text{core}}$ such that the algorithm setting a fixed threshold of $\tau_{\text{core}}$ obtains $\E[\ALG] \geq \Omega(V)$.
\end{claim}
\begin{proof}
  Let $Y_j'$ be a copy of $Y_j$ conditioned on falling into the range $[0,\tau]$.
  Let $X_i' = \sum_{j \in S_i} Y_j'$ and let $W = \max_i X_i'$.
  We have $\E[\max_i X_i \mid A=\emptyset] = \E[W] = V$.

  Now, note that $W$ is an XOS function of the independent $Y_i'$ variables (meaning is a maximum of weighted combinations).
  Thus we can apply the concentration of XOS functions (more generally, for self-bounding functions, see e.g.~\cite{Vondrak-arXiv10}) to get
\[		\Pr\big[W < (1 - \delta)\cdot \E[W] \big] \leq  \exp\Big(-\delta^2 \cdot \frac{\E[W]}{2\tau} \Big).
\]
In particular, for $\delta=1/2$ we get
\[ 	\Pr\Big[W <\frac12 \E[W] \Big] \quad \leq \quad \exp\Big(-\frac{\E[W]}{8\tau} \Big) \quad \leq \gamma,
\]
for some constant $\gamma < 1$, using that $\tau \leq \E[W] / 10$.
Hence the expected value of an algorithm that sets a threshold of $\frac{1}{2} \E[W] $ is at least
\[ \Pr\Big[W \geq\frac{1}{2} \E[W] \Big] \cdot \frac{1}{2} \E[W]  \geq   \Omega\Big(\E[W]\Big).			\qedhere
\]
\end{proof}

Now we have all the tools to prove our main theorem.

\begin{proof}[Proof of Theorem~\ref{thm:unwtdLinCorr}] 
  By Claim \ref{claim:unwtdBoundOPT}, one of the follwing is at least a ${2}$-approximation to the prophet: $V := \E[\max_i X_i \mid A=\emptyset]$, and $\E[\sum_j p_j \cdot \E[Y_j \mid Y_j > \tau]$.
  Suppose it is the latter.
  Then by Claim~\ref{claim:unwtdMaxYj}, setting a fixed threshold of $\tau_{\text{tail}}$ gives a constant-factor approximation.

  So suppose we have $V \geq \E[\OPT]/2$.
  If $\tau > V/10$, then we can set a fixed threshold of $\tau$: with probability at least $\frac{1}{2}$, some $X_i \geq \tau$ (because some $Y_j \geq \tau$), so we obtain performance at least $\frac{1}{2}\tau \geq \frac{1}{40}\E[\OPT]$.

  Finally, if $V \geq \E[\OPT]/2$ and $\tau \leq V/10$, then by Claim \ref{claim:unwtdApplyConcent}, setting a fixed threshold of $\tau_{\text{core}}$ gives a constant-factor approximation.
\end{proof}

\section{Negatively Associated Values} \label{sec:negCorr}
In this section we show  a ${2}$ approximation ratio for \emph{negatively associated} random values, a property that is known to imply negative correlation\footnote{We say $\{X_j\}$ are negatively correlated if for all $i,i'$ we have $\text{Cov}(X_i,X_{i'})\leq 0$.  }~\cite{JP-Journal83}. 
Formally, we say $\{X_j\}$ are negatively associated  
if for all monotone increasing functions $f,g$ and disjoint subsets $S,S' \subseteq \{1,\dots,n\}$, and for all $a,b \in \reals$, we have
  \[ \Pr\Big[ f(X_j : j \in S) \geq a\Big] \leq \Pr\Big[ f(X_j : j \in S) \geq a \mid g(X_j : j \in S') \leq b \Big] . \]

Let $\tau = \frac{1}{2}\E[\max_i X_i]$ and let $P = \Pr[\max_i X_i \geq \tau]$.
We can write down precisely the usual prophet proof with just one line requiring additional justification.
\begin{align*}
  \E[\ALG_{\tau}]
  &=    \textstyle P \cdot \tau + \sum_{i=1}^n \Pr[X_{i'}<\tau (\forall i'<i)] \cdot \E\left[(X_i - \tau)^+ \mid  X_{i'} < \tau (\forall i' < i)\right]  \\
  & \textstyle  \geq P\cdot \tau  + (1-P) \cdot \sum_{i=1}^n \E\left[(X_i-\tau)^+ \mid X_{i'} < \tau (\forall i'<i)\right]  \\
  &\textstyle  \geq P\cdot \tau + (1-P) \cdot \sum_{i=1}^n \E\left[(X_i -\tau)^+\right] ,
\end{align*}
where the last inequality uses negative association as $(X_i-\tau)^+$ is a monotone function as is $\max_{i'<i} X_{i'}$. We can also use weaker notions of negative correlation for the last inequality, e.g., NLODS as in Section~2 of~\cite{RS-Journal92}.
Now repeating the old prophet inequality analysis,  
\[ \sum_{i=1}^n \E\Big[(X_i -\tau)^+\Big] \quad \geq \quad \E\Big[ \sum_{i=1}^n (X_i -\tau)^+\Big]
\quad \geq \quad \E\Big[ \max_i X_i -\tau\Big]  \quad= \quad \tau.
\]
Thus, 
  $\E[\ALG_{\tau}] \geq  P\cdot \tau + (1-P) \cdot \tau  
  =    \tau .$


\section{Bounded $\scol$ and Small Cardinality Constraint} \label{sec:uniform-constant}
For the $r$-uniform matroid problem with bounded $\scol$, we have shown in Section \ref{subsec:uniform-scol} a $1+ o(1)$ approximation for large $r$ tending to infinity.
Here, we complement that result with a gracefully-improving approximation ratio for all $r$ that smoothly interpolates between $O\left({\scol}\right)$ for $r=1$ (the classic result) and $O(1)$ for $r \geq \scol$.

The approach is an extension of our algorithm for bounded column sparsity in the $r=1$ case.

\begin{theorem} \label{thm:constSmallScol}
  For any $\scol, r$, the linearly correlated prophets problem with column sparsity $\scol$ and cardinality constraint $r$ admits an approximation ratio of ${2e^2} \cdot \max\left\{ 1, \frac{\scol}{r} \right\}$.
\end{theorem}
In other words, as $r=2,3,\ldots,\scol$, the guarantee improves to a constant factor times $\frac{2}{\scol}, \frac{3}{\scol}, \ldots, 1$.

\begin{proof}
  Let there be $r$ sets (``buckets'') $B_1,\dots,B_r$.
  If $r < \scol$, let there also be a ``discard pile'' $B_0$.
  Let $c = \max\{r,\scol\}$. 
  
  For each $X_i$, place it in a bucket $j \in \{1,\dots,r\}$ each chosen with probability $\frac{1}{c}$.
  If $r < \scol$, then with the remaining probability of $\frac{\scol-r}{\scol}$, place $X_i$ in the discard bucket $B_0$.

  For each bucket $j=1,\dots,r$, give the bucket a cardinality constraint of $1$ item and run the following algorithm (based on the $r=1$ case).
  Let $S_j = \{i : X_i \in B_j\}$ and note that they are disjoint for different $j$.
  When a variable $X_i$ arrives, send it to the algorithm for its bucket, or if it is in $B_0$, discard $X_i$ and continue.
  In bucket $j$, we run an inclusion-threshold algorithm with $S_j$ and with $\tau_j$ to be determined next.
  Assign each $Y_{j'}$ to the first $X_i \in B_j$ that includes it, i.e., let $T_i = \{j' : A_{ij'} > 0 \text{ and } A_{i'j' = 0} (\forall i' < i, i' \in S_j)\}$.
  Let $Z_i = \sum_{j' \in T_i} A_{ij'} Y_{j'}$, and let $\tau_j = \frac{1}{2}\E[ \max_{i\in S_j} Z_i]$.

\begin{claim} For each bucket $j \in \{1,\dots,r\}$, the expected value selected by the algorithm is at least $\frac{1}{2e} \cdot \E[\max_{i\in S_{j}} X_i]$.
\end{claim}
\begin{proof}
  By construction each $Z_i$ is independent of all previous $X_i$ in the same bucket, so by the Augmentation Lemma~\ref{lemma:aug}, bucket $j$ obtains expected reward at least $\frac{1}{2}\E[ \max_{i \in S_j} Z_i]$ with randomness over the variables.

  For each $Y_{j'}$ with $A_{ij'} > 0$, we claim $\Pr[j \in T_i] \geq \frac{1}{e}$ because there are at most $\scol-1$ other variables $X_{i'}$ that include $Y_{j'}$, and each misses bucket $j$ with probability at least $1-\frac{1}{\scol}$, so they all miss bucket $j$ with probability at least $(1-\frac{1}{\scol})^{\scol-1} \geq \frac{1}{e}$.
  In this case, we must have $j \in T_i$.
  So for each fixed $X_i$, we have with probability only over the bucket assignments and construction of $S_j$,
  \begin{align*}
    \E[ Z_i]
      \quad = \quad    \sum_{j'} \Pr[j' \in T_i]\cdot A_{ij'}Y_{j'}  
      \quad \geq \quad \frac{1}{e} \sum_{j'} A_{ij'} Y_{j'}  
      \quad= \quad    \frac{1}{e} X_i .
  \end{align*}
  Combining these facts, each bucket $j$ obtains expected reward at least $\frac{1}{2e}\E[ \max_{i \in S_j} X_i]$.
  \end{proof}
  
  Write $\E_B$ for an expectation taken over the bucketing and $\E_{\X}$ for expectation over the realizations of the variables.
  The above gives that for every set of variable realizations, $\E_B[\ALG] \geq \frac{1}{2e}\sum_{j=1}^r \E_B [\max_{i \in S_j} X_i]$.
  By linearity of expectation and symmetry of the buckets, we have
  \begin{align}
    \E_{\X,B} [\ALG] &\geq \frac{r}{2e} \E_{\X,B} \big[ \max_{i \in S_1}X_i \big]. \label{eqn:matroid-col-ineq}
  \end{align}

\begin{claim}  \label{claim:constColS}
$\E_{\X,B} \big[ \max_{i \in S_1}X_i \big] \geq \frac{1}{e\cdot c} \cdot \E_{\X}[\OPT]$, where $c = \max\{r,\scol\}$.
\end{claim}
\begin{proof}
  Let the random variable $X^{(i)}$ equal the $i$th-largest realized variable, i.e., in particular
  \begin{align*}
    \OPT &= \sum_{i=1}^r X^{(i)} .
  \end{align*}
  Recall that any fixed variable falls into bucket $B_1$ with probability $\frac{1}{c}$.
  Fixing realizations of $X^{(1)},\dots,X^{(r)}$, we have with probability taken only over the buckets,
  \begin{align*}
    \E_B \big[\max_{i\in S_1} X_i \big]
    &\geq \sum_{i=1}^r \Pr\left[X^{(i)} \in B_1 \text{ and } X^{(i')} \not\in B_1 (\forall i'<i)\right] X^{(i)}  \\
    &=    \sum_{i=1}^r \frac{1}{c}\left(1-\frac{1}{c}\right)^{i-1} X^{(i)} 
    \quad \geq \quad \frac{1}{e \cdot c} \sum_{i=1}^r X^{(i)}  
    \quad = \quad   \frac{1}{e \cdot c} \OPT . 
  \end{align*}  
  Now taking an expectation on both sides over the realizations of  $\X$ proves the claim.
  \end{proof}
 Finally, combine Claim~\ref{claim:constColS} with Inequality~\eqref{eqn:matroid-col-ineq} to get
  \begin{align*}
    \E[\ALG]
        \quad \geq \quad \frac{r}{2 e} \E_{\X} \E_B \big[ \max_{i \in S_1} X_i \big] 
        \quad \geq \quad \frac{r}{2e^2 c} \E_{\X} [ \OPT ] , 
  \end{align*}
  which proves Theorem~\ref{thm:constSmallScol}.
\end{proof}


\section{Missing Proofs} \label{app:proofs}

\subsection{Missing Proofs from Section~\ref{sec:model}}
\fixedThresh*
\begin{proof}[Proof of Lemma~\ref{thm:fixedThreshold}]
  We consider the $\srow=\scol=2$ tower instance (Example \ref{ex:2-tower}), with sufficiently small $\epsilon$ chosen later.
  We claim that $\E [\max_i X_i] = \Omega(n)$ while any fixed threshold-$\tau$ algorithm has $\E [\ALG_{\tau}] \leq 3$.
  
  First, we bound $\E [\ALG_{\tau}]$.
  Let $p_j = \Pr[\text{$\ALG_{\tau}$ takes $X_j$ and $Y_j$ is active}]$.
  In this case the algorithm's reward includes $Y_j = \frac{1}{\epsilon^j}$ with coefficient $1$.
  Let $p_j' = \Pr[\text{$\ALG_{\tau}$ takes $X_{j-1}$ and $Y_j$ is active}]$.
  In this case its reward includes $Y_j = \frac{1}{\epsilon^j}$ with coefficient $\epsilon$.
  We note that $p_j,p_j' \leq \Pr[\text{$Y_j$ is active}] = \epsilon^j$.
  By summing over $Y_1,\dots,Y_n$, we have
  \begin{align*}
    \E[\ALG_{\tau}] \quad =  \quad  \sum_{j=1}^n \Big( p_j \frac{1}{\epsilon^j} + (p_j') (\epsilon) \frac{1}{\epsilon^j} \Big) \quad \leq  \quad \sum_{j=1}^n \Big( p_j \frac{1}{\epsilon^j} + \epsilon \Big)
    \quad=\quad    n\epsilon + \sum_{j=1}^n p_j \frac{1}{\epsilon^j} .
  \end{align*}
  We argue that $p_j = 0$ for all but at most two terms.
  Let $j^*$ satisfy ${\epsilon^{-(j^*+1)}} < \tau \leq {\epsilon^{-j^*}}$, or $j^* = 1$ if $\tau \leq \frac{1}{\epsilon}$.
  We claim $p_j = 0$ if $j \leq j^*-2$: assuming $\epsilon < \frac{1}{2}$, we have 
  \[ X_j \quad \leq \quad \frac{2}{\epsilon^j} \quad \leq \quad \frac{1}{\epsilon^{j+1}} \quad < \quad \tau, \]
   so $X_j$ is never taken.
  We also claim $p_j = 0$ if $j \geq j^*+1$: If $Y_j$ is active, then 
  \[ X_{j-1} \quad \geq \quad \epsilon \frac{1}{\epsilon^j} \quad \geq \quad \frac{1}{\epsilon^{j-1}} \quad \geq \quad \tau,\]
   so $X_{j-1}$ is taken and $X_j$ is not.
  So we have $p_j = 0$ unless $j \in \{j^*-1, j^*\}$, in which case $p_j \leq \epsilon^j$.
  So
  \begin{align*}
    \E[\ALG_{\tau}] \quad \leq \quad n\epsilon + \sum_{j=1}^n p_j \frac{1}{\epsilon^j} 
               \quad \leq \quad n\epsilon + 2 
               \quad \leq \quad 3 
  \end{align*}
  for $\epsilon \leq \frac{1}{n}$.

  For the benchmark, since $\max_i\{X_i \} \geq \max_j\{Y_j\} $, it suffices to show $\E[\max_j Y_j] = \Omega(n)$. Now,
    \begin{align*} 
  \E[\max_j\{Y_j\}]  &=   \textstyle \sum_{i=1}^n \Pr[Y_j = 0 ~ (\forall j > i)] \cdot \Pr[Y_i \neq 0] \cdot \Big(\frac{1}{\epsilon^i}\Big)  \\
      &\textstyle \geq \Pr[Y_j = 0 ~ (\forall j)] \cdot \sum_{i=1}^n  \Pr[Y_i \neq 0] \cdot \Big(\frac{1}{\epsilon^i}\Big).
        \end{align*}
  Since $\Pr[Y_i \neq 0] = \epsilon^i$, we get
        \begin{align*}
  \E[\max_j\{Y_j\}]    \quad &\geq \quad    \Pr[Y_j = 0 ~ (\forall j)] \cdot n  \\
  &\geq \quad \textstyle \big(1 - \sum_j \Pr[Y_j \neq 0] \big) \cdot n \quad \geq \quad \left(1 - n\epsilon\right) \cdot n  \quad \geq \quad {n}/{2}
  \end{align*}
  for any choice of $\epsilon \leq \frac{1}{2n}$.
  This gives an approximation ratio of at least $\frac{n/2}{3} = \frac{n}{6}$ for $\epsilon \leq \frac{1}{2n}$.
\end{proof}

\subsection{Missing Proofs from Section~\ref{sec:multipItemsAug}}
\begin{proof}[Proof of Claim~\ref{claim:onlyOneLarge}]
  First, consider supplementing $\OPT'$ by including tiny elements below $\tau_c$ when there is room, i.e. let $Z^{(i)\prime\prime} = Z^{(i)} \one[Z_i < \tau_0]$ and consider $\OPT_2'' = \sum_{j=1}^r Z^{(j)\prime\prime}$.
  Let $\OPT'' = \OPT_1' + \OPT_2''$.
  On a case-by-case basis, $\OPT''$ differs from $\OPT'$ by at most $r$ elements, each at most $\tau_c \leq \frac{\epsilon}{r} \E[\OPT]$.
  This proves that $\E[\OPT'] \geq \E[\OPT''] - \epsilon \E[\OPT]$.
  We next prove that $\E[\OPT''] \geq (1-\epsilon)\E[\OPT]$, which completes the proof of the claim.

  Let $H$ be the event there exists a heavy element, i.e. $\max_i Z_i \geq \tau_0 = \mathbb{E}[\OPT]/\epsilon$.
  Let $p = \Pr[H] $. So,
  \begin{align}
    \E [\OPT''] &=    p \cdot \E[\OPT'' \mid H] + (1-p) \cdot \E[\OPT'' \mid \lnot H] \nonumber  \\
             &=    p \cdot \E[\OPT'' \mid H] + (1-p) \cdot \E[\OPT \mid \lnot H]  \nonumber  \\
             &\geq p \cdot \E[Z^{(1)} \mid H] + (1-p) \cdot \E[\OPT \mid \lnot H] . \label{eqn:opt-prime-exceeds}
  \end{align}
  Now, we claim $\E[\OPT \mid H] \leq \E[Z^{(1)} \mid H] + \E[\OPT]$.
  Proof: let $M_i$ be the event that $i = \arg\max_{i'} Z_{i'}$ and $Z_i \geq \tau_0$.
  Let $\OPT_{-i}$ be the sum of the largest $r-1$ elements excluding $Z_i$.
  Note that conditioning on all others lying below $Z_i$, for any fixed  $Z_i$, only decreases $\OPT_{-i}$, as the variables are independent.
  \begin{align*}
    \E[\OPT \mid H] &= \E[Z^{(1)} \mid H] + {\textstyle \sum_{i=1}^n } \Pr[M_i \mid H] \cdot \E[\OPT_{-i} \mid M_i]  \\
                    &\leq \E[Z^{(1)} \mid H] + {\textstyle \sum_{i=1}^n }  \Pr[M_i \mid H] \cdot \E[\OPT_{-i}]  \\
                    &\leq \E[Z^{(1)} \mid H] + {\textstyle \sum_{i=1}^n }  \Pr[M_i \mid H] \cdot \E[\OPT]  \quad = \quad   \E[Z^{(1)} \mid H] + \E [\OPT ].
  \end{align*}
  Using this,
  \begin{align*}
    \E [\OPT] &=    p \cdot \E[\OPT \mid H] + (1-p) \cdot \E[\OPT \mid \lnot H]  \\
            &\leq p \cdot \E[Z^{(1)} \mid H] + p \cdot \E [\OPT] + (1-p) \cdot \E[\OPT \mid \lnot H] .            
  \end{align*}
  This implies
  \begin{align*}
     (1-p)\E[\OPT] ~~\leq ~~ p \cdot \E[ Z^{(1)} \mid H] + (1-p) \cdot \E[\OPT \mid \lnot H] .
  \end{align*}
  Combining with Inequality \eqref{eqn:opt-prime-exceeds} gives $\E[\OPT''] \geq (1-p)\E[\OPT]$.
  We have $p \leq \epsilon$ by Markov's inequality: $p = \Pr[\max_i Z_i \geq \tau_0] \leq \E[\max_i Z_i]/\tau_0 \leq \E[\OPT]/\tau_0 = \epsilon$.
  This completes the proof.
\end{proof}

\begin{proof}[Proof of Lemma~\ref{lemma:alg-approx-algprime}]
  Suppose $r$ is large enough that $\epsilon \leq 0.5$, otherwise the lemma is immediate.
  Let $\delta := \frac{1 + c \cdot \beta}{r}$.
  \begin{claim} \label{claim:eps-small-discard}
    $\epsilon \geq 2\delta$.
  \end{claim}
  \begin{proof}
    Using that $\epsilon \geq r^{-1/4}$, we have $c := \lceil \frac{1}{\epsilon} \ln \frac{r}{\epsilon^2} \rceil \leq \frac{3}{\epsilon} \ln(r)$.
    Further using that $\epsilon \geq 3r^{-1/4}$, this implies $\ln\frac{c}{\epsilon} \leq \ln\left(r^{1/2} \ln(r)\right) \leq \ln(r)$.
    Therefore, $\beta := 3 \sqrt{r \ln\frac{c}{\epsilon}} \leq 3 \sqrt{r \ln(r)}$.
    Then $c \cdot \beta \leq \left(\frac{3}{\epsilon}\ln(r)\right) \left(3 \sqrt{r \ln(r)}\right) \leq \frac{9 \sqrt{r} \left(\ln r\right)^{3/2}}{\epsilon}$.
    Using that $\epsilon \geq 9r^{-1/4}\left(\ln r\right)^{3/2}$, this gives $c \cdot \beta \leq r^{3/4}$, so $1 + c \cdot \beta \leq 2r^{3/4}$, so $\delta \leq 2 r^{-1/4} \leq \epsilon / 2$.
  \end{proof}
  Let $K$ be the number of arrivals taken by $\ALG$, a random variable.
  \begin{claim}
    With  at least $1-\epsilon$ probability, $K < r$ (i.e. $\ALG$ does not reach its cardinality constraint).
  \end{claim}
  \begin{proof}
    The number of arrivals taken by $\ALG'$ is at most $\sum_{j=0}^c \tilde{r}_j = 1 + c \cdot \beta + \sum_{j=1}^c r_j$.
    Because $\OPT$ takes at most $r$ arrivals pointwise, and thus in expectation, we have $\sum_{j=1}^c r_j \leq r$.
    So $\ALG'$ takes at most, in the worst case, $K' = r + 1 + c \cdot \beta = r(1+\delta)$ arrivals.
    Because $\ALG$ keeps each independently with probability $1-\epsilon \leq 1-2\delta < (1-\delta)^2$, the chance it reaches $r$ is upper-bounded by the chance that a Binomial($K',(1-\delta)^2$) variable exceeds $r$.
    This is upper-bounded by the chance it exceeds $K'(1-\delta)^2(1+\delta) = r(1-\delta)^2(1+\delta)^2 = r(1-\delta^2)^2 < r$.
    So by a Chernoff bound,
    \begin{align*}
      \Pr[K \geq r] \quad \leq \quad \Pr[\text{Binom}(K',(1-\delta)^2) \geq K'(1-\delta)^2(1+\delta)]  
                    \quad \leq \quad \exp\left(\frac{-\delta^2 K'(1-\delta)^2}{3}\right).
    \end{align*}
    We have $\delta \leq \epsilon/2 \leq 0.25$, and $K' \geq r$, so $K'(1-\delta)^2 \geq \frac{r}{2}$.
    Also, $\delta \geq \frac{c \cdot \beta}{r} \geq \frac{\beta}{r}$.
    \begin{align*}
      \Pr[K \geq r] &~\leq~ \exp\left(\frac{-\delta^2 r}{6}\right) 
                    ~\leq~ \exp\left(\frac{-\beta^2}{6 r}\right)  
                    ~\leq~ \exp\left(-\ln \frac{c}{\epsilon}\right)  
                    ~\leq~ \epsilon .	\qedhere
    \end{align*}
  \end{proof}
  
  Now, each time $\ALG'$ obtains some variable $X_i$, $\ALG$ also obtains it unless either: it has reached its cardinality constraint; or it independently discards $X_i$ (with probability $\epsilon$).
  By a union bound over these two events, when $\ALG'$ obtains $X_i$, $\ALG$ also obtains it except with probability $1-2\epsilon$.
  So $\E[\ALG] \geq (1-2\epsilon)\E[\ALG']$, which proves Lemma~\ref{lemma:alg-approx-algprime}.
\end{proof}

{\small
\bibliographystyle{alpha}
\bibliography{bib}
}

\end{document}